\documentclass[journal]{IEEEtran}
\usepackage{amsmath,amsfonts}
\usepackage{amssymb}
\usepackage[linesnumbered,ruled,vlined]{algorithm2e}
\usepackage{amsthm}
\newtheorem{theorem}{Theorem}
\newtheorem{lemma}{Lemma}
\newtheorem{assumption}{Assumption}
\newtheorem{corollary}{Corollary}

\theoremstyle{remark}
\newtheorem{remark}{Remark}

\usepackage{bm}
\usepackage{array}
\usepackage{graphicx}
\usepackage{subfigure}
\usepackage{textcomp}
\usepackage{stfloats}
\usepackage{url}
\usepackage{verbatim}
\usepackage{pifont}
\usepackage{cite}
\usepackage{tablefootnote}
\usepackage{makecell}
\usepackage{multirow}
\usepackage{booktabs}
\usepackage{float}
\usepackage[acronym]{glossaries}
\newcommand{\newac}{\newacronym}
\newcommand{\ac}{\gls}

\newcommand{\acpl}{\glspl}

\newac{nlos}{NLOS}{Non-Line-of-Sight}
\newac{iot}{IoT}{Internet of Things}
\newac{gnss}{GNSS}{Global Navigation Satellite Systems}
\newac{rf}{RF}{Radio Frequency}
\newac{toa}{TOA}{Time-of-Arrival}
\newac{aoa}{AOA}{Angle-of-Arrival}
\newac{knn}{KNN}{K-Nearest Neighbors}
\newac{los}{LOS}{Line-of-Sight}
\newac{csi}{CSI}{Channel State Information}
\newac{ap}{AP}{Access Point}
\newac{agc}{AGC}{Automatic Gain Control}
\newac{rssi}{RSSI}{Received Signal Strength Indicator}
\newac{tof}{TOF}{Time-of-Flight}
\newac{mmse}{MMSE}{Minimum Mean Square Error}
\newac{vae}{VAE}{Variational Autoencoder}

\begin{document}
\title{Environment-Conditioned Diffusion Meta-Learning for Data-Efficient WiFi Localization}

 \author{Jun Gao, Zheng Xing, Wenliang Lin, Weibing Zhao, Xuhui Zhang,\\ Junting Chen, Zhongliang Deng, Shuguang Cui%
\thanks{J. Gao, Z. Xing, W. Zhao, and X. Zhang are with the Shenzhen Future Network of Intelligence Institute, The Chinese University of Hong Kong, Shenzhen, China (e-mail: gaojun@cuhk.edu.cn; zhengxing@link.cuhk.edu.cn; weibingzhao@msubit.edu.cn; zhangxuhui@cuhk.edu.cn).}%
\thanks{W. Lin and Z. Deng are with the School of Electronic Engineering, Beijing University of Posts and Telecommunications, Beijing, China (e-mail: charterlin@bupt.edu.cn; zld2026@163.com).}%
\thanks{J. Chen is with the Shenzhen Future Network of Intelligence Institute and the School of Science and Engineering, The Chinese University of Hong Kong, Shenzhen, China (e-mail: juntingc@cuhk.edu.cn).}%
\thanks{S. Cui is with the School of Science and Engineering (SSE), the Future Network of Intelligence Institute (FNii), and the Guangdong Provincial Key Laboratory of Future Networks of Intelligence, The Chinese University of Hong Kong; he is also affiliated with Peng Cheng Laboratory, Shenzhen, China (e-mail: shuguangcui@cuhk.edu.cn).}%
}

\maketitle

\begin{abstract}	

Fingerprinting-based WiFi localization suffers from cross-environment distribution shifts, particularly when only a small number of labeled samples are available in the target environment. Existing adaptation methods mainly rely on wireless measurements and rarely exploit the physical propagation geometry that governs CSI fingerprints. This paper proposes EnvCoLoc, an environment-conditioned diffusion meta-learning framework for data-efficient WiFi localization. EnvCoLoc extracts material-aware descriptors from registered 3D point clouds within the first Fresnel zone and uses them to condition a latent stochastic generator. The generator produces environment-specific parameter offsets to modulate a shared meta-learned initialization, enabling the localization network to start adaptation from a geometry-aware parameter region. The generator and localization network are jointly optimized under a two-loop meta-learning objective, so that environmental priors are injected before support-set fine-tuning while residual task-specific mismatch is corrected by gradient adaptation. We further provide a conditional local excess-loss analysis for finite-step adaptation. The bound decomposes the adaptation error into initialization quality, environment-conditioned residual variation, stochastic generation errors, and inner-loop contraction, thereby explaining when geometric conditioning and stochastic offset generation improve the local adaptation bound. Experiments using commodity WiFi devices and Leica RTC360 LT point-cloud acquisition show that EnvCoLoc consistently improves few-shot localization performance. With only 10 target support samples, EnvCoLoc reduces the mean localization error by 26.1\% in the LOS lab scenario and 24.9\% in the NLOS office scenario compared with MetaLoc.
\end{abstract}
\begin{IEEEkeywords}
Wireless localization, diffusion model, meta-learning, fingerprinting.
\end{IEEEkeywords}
\section{Introduction}

The pervasive deployment of spatially-aware \ac{iot} applications has made high-precision indoor localization a fundamental requirement for next-generation wireless networks~\cite{chen2022fidora,xing2024constructing,wang2020csi,10230036,yuan2024risfedbroad}. Since \ac{gnss} suffer from intrinsic signal blockages in indoor environments~\cite{7275552}, \ac{rf}-based positioning has emerged as a highly viable alternative. Although conventional geometric methods relying on \ac{toa} or \ac{aoa} can achieve centimeter-level accuracy~\cite{li2021transloc,feng2020kalman,9737357}, they exhibit severe performance degradation under \ac{nlos} conditions~\cite{10297296} and necessitate strict temporal synchronization or complex signal decoding with wireless devices. This challenge becomes more pronounced as emerging B5G/6G systems increasingly exploit high-directionality and spatially selective wireless links~\cite{lin2024singleboard}. Fingerprinting-based localization effectively overcomes these challenges, offering a cost-effective deployment strategy without the need for specialized ranging hardware or stringent synchronization constraints.

Traditional fingerprinting methods typically operate in two distinct phases: establishing an offline fingerprint database that links spatial coordinates to wireless features, and subsequently matching online real-time measurements against this database to estimate the user's position. Existing fingerprinting methodologies have progressively evolved across three primary paradigms: deterministic, probabilistic, and deep learning-based approaches. Deterministic methods, such as \ac{knn} variants~\cite{bahl2000radar,wu2017gain,he2014sectjunction,jiang2012ariel}, prioritize computational efficiency by directly comparing real-time signals with pre-stored fingerprints using simple similarity metrics. Yet, their reliance on rigid, point-to-point matching renders them highly sensitive to signal noise and incapable of adapting to environmental shifts. To enhance robustness against signal fluctuations, probabilistic and statistical methods~\cite{youssef2005horus,4536547} infer locations by modeling the uncertainty of wireless fingerprints. While this principle explicitly captures signal uncertainty, continually updating these complex statistical models to reflect dynamic environmental changes is notoriously inefficient. Consequently, recent efforts have shifted toward deep learning-based methods~\cite{wang2015deepfi,ghozali2019indoor,chen2017confi,hsieh2019deep}, which leverage advanced neural architectures to automatically extract generalized, robust feature representations from complex wireless measurements. Although these data-driven advancements demonstrate powerful representation capabilities, they still confront two critical bottlenecks: adapting efficiently to new environments and maintaining high localization accuracy when only a sparse set of labeled samples is available.

To address the challenge of environmental adaptation, recent deep learning-based fingerprinting studies have proposed solutions from both algorithmic and data-centric perspectives. On the algorithmic side, domain adaptation and continual learning techniques are widely employed to mitigate distribution shifts between historical and new environments. For instance, TransLoc~\cite{li2021transloc} utilizes transfer learning to map cross-domain features into a shared homogeneous space, while CRISLoc~\cite{CRISLoc} reconstructs outdated \ac{csi} databases using a sparse set of new measurements. Other approaches leverage generative models, such as the variational autoencoder in Fidora~\cite{chen2022fidora}, or employ incremental learning, like ILCL~\cite{zhu2022intelligent}, to update models with reduced training overhead. On the data side, efforts have focused on engineering robust signal representations to combat environmental noise. ViVi~\cite{wu2017gain}, for example, mitigates \ac{rssi} uncertainty by exploiting spatial gradients, whereas CiFi~\cite{cifi} and DFPS~\cite{fang2015novel} enhance \ac{csi} stability by extracting phase and amplitude differences between antenna pairs. However, despite these advancements, existing methods rely almost entirely on mining wireless signals, inherently neglecting the physical environmental geometry that dictates signal variations.
\begin{figure}[!t]
	\centering
	\includegraphics[width=1\columnwidth]{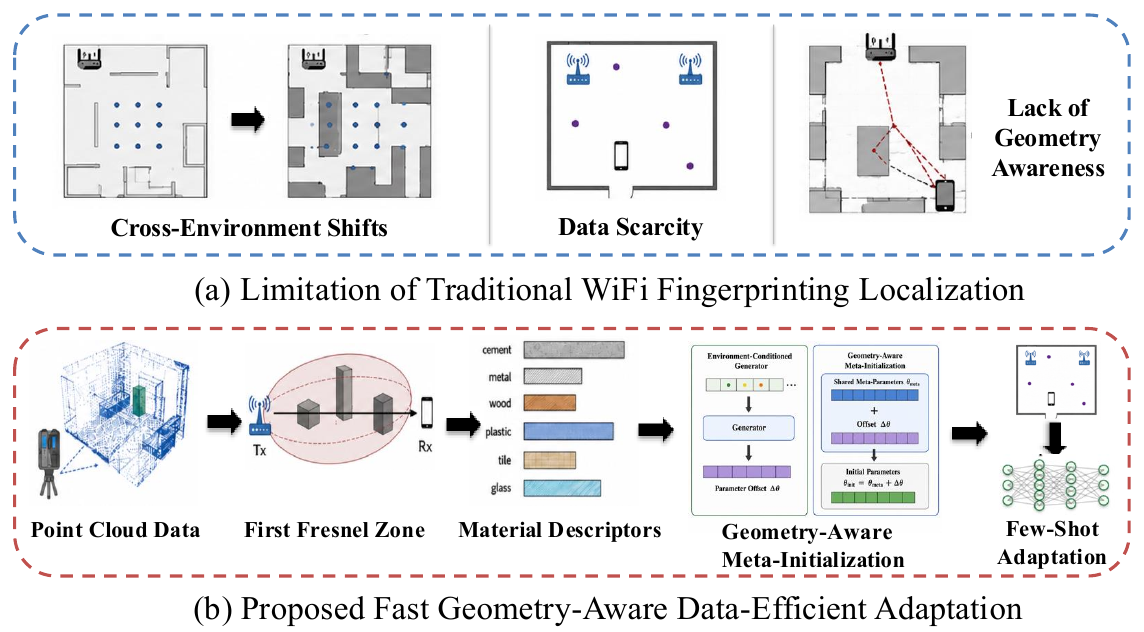}
	\caption{Overview of the limitation in traditional WiFi fingerprinting and the proposed method for fast geometry-aware data-efficient adaptation.}	\label{fig:background}
\end{figure}

The physical distribution of obstacles between a transmitter and a receiver governs multipath propagation, which in turn defines the unique \ac{csi} fingerprint at any given location. Recognizing this, recent methodologies have increasingly incorporated explicit environmental data, such as simulated ray-tracing~\cite{9354041} and point-cloud-based environmental information~\cite{10757755}, to model propagation environments and reduce fingerprint ambiguity. Related progress in reconfigurable intelligent surfaces also indicates that propagation behavior can be strongly shaped by environmental structure~\cite{lu2025risdesign}. Parallel to these environmental integrations, our prior work, MetaLoc~\cite{gao2022metaloc,10274764}, pioneered gradient-based meta-learning to enable rapid model adaptation in wireless localization. While a natural evolution to overcome MetaLoc's lack of environmental awareness is to condition meta-initialization on geometric descriptors using deterministic mapping networks, including hypernetworks~\cite{ha2016hypernetworks} and MLP-based structures~\cite{tolstikhin2021mlpmixer}, this approach exhibits a critical limitation. Specifically, deterministic conditioning yields only a single, point-estimate parameter correction for any given environmental descriptor. This deterministic approach can be insufficient because sub-regions sharing similar coarse environmental features may experience distinct local multipath variations, making distribution-based parameter corrections more suitable than fixed outputs.

As illustrated in Fig.~\ref{fig:background}, to fully exploit the interplay between environmental geometry and wireless propagation, this paper proposes integrating geometry-aware priors into a meta-learning framework tailored for fingerprint-based localization. The key idea is to utilize structured descriptors derived from 3D point clouds to provide an environment-specific parameter initialization for the fingerprinting model. However, realizing this geometry-guided fingerprinting approach requires overcoming two primary technical challenges:
\begin{itemize}
	\item \textbf{Stochastic Geometry-to-Parameter Mapping:} It is highly challenging to accurately model the complex and stochastic relationship between coarse environmental features and the high-dimensional network parameter shifts required for robust model adaptation.
	\item \textbf{Stable Joint Optimization under Limited Data:} Jointly optimizing an environment-dependent generator alongside the localization meta-parameters requires a robust training strategy to prevent overfitting and ensure convergence when only sparse support data is available.
\end{itemize}

In this paper, we propose EnvCoLoc, an environment-conditioned diffusion meta-learning framework for data-efficient WiFi localization. To overcome the limitations of deterministic conditioning, a generative model capable of learning a conditional distribution over parameter corrections is required. Among generative architectures, Denoising Diffusion Probabilistic Models (DDPMs)~\cite{ho2020denoising,rombach2022high} are particularly suitable because they provide a stable stochastic denoising-style parameterization for modeling complex conditional distributions without requiring an invertible generator. In EnvCoLoc, this generator is optimized end-to-end through the meta-learning objective rather than through a supervised noise-prediction loss. EnvCoLoc bridges the gap between physical propagation geometry and wireless signals by leveraging 3D point-cloud descriptors to condition a latent diffusion model. This generative model produces environment-specific parameter offsets that modulate a shared meta-initialization, enabling rapid and robust model adaptation to new environments with minimal labeled data. In summary, we make the following key contributions:
\begin{itemize}
	\item \textbf{Geometry-Aware Meta-Initialization:} We introduce a meta-initialization mechanism that uses point-cloud descriptors and a conditional latent diffusion model to generate task-specific parameter offsets. This explicitly captures environment-dependent propagation characteristics, providing a geometry-aware initialization for rapid adaptation.
	
	\item \textbf{Joint Diffusion Meta-Learning Optimization:} We develop a two-loop optimization framework to jointly learn the diffusion offset generator and the shared meta-parameters. Modeling the offset as a conditional distribution accommodates local multipath uncertainties prior to few-shot inner-loop fine-tuning.
	
	\item \textbf{Theoretical Convergence Analysis:} We provide a conditional local excess-loss analysis that decomposes finite-step adaptation error into initialization quality and inner-loop contraction. The analysis explains how environment-conditioned initialization can tighten the adaptation bound and specifies when stochastic generation can improve over a deterministic offset generator.
	
	\item \textbf{Real-World Validation:} Using a custom commodity-device WiFi platform, we collected paired \ac{csi} and 3D point-cloud data in both \ac{los} and \ac{nlos} environments. With only 10 support samples, EnvCoLoc reduces the mean localization error by 26.1\% (\ac{los}) and 24.9\% (\ac{nlos}) compared to the MetaLoc framework. Ablation results show that removing environment conditioning or replacing the diffusion generator with a deterministic mapping both degrade accuracy, confirming that both components contribute to the overall performance gain.
\end{itemize}

The remainder of this paper is organized as follows. Sec.~\ref{sec:FL} introduces the system model and data preprocessing pipeline. Sec.~\ref{sec:ECML} details the proposed EnvCoLoc framework. Sec.~\ref{sec:theory} provides the theoretical convergence analysis. Sec.~\ref{sec:ER} presents the experimental setup and comprehensive performance evaluation. Finally, Sec.~\ref{sec:conclusion} concludes the paper.

\section{System Scenario and Data Representation}
\label{sec:FL}
This section establishes the foundation for the proposed localization framework. We first describe the physical system scenario and the wireless channel model in Sec.~\ref{section:scenario}. Subsequently, Sec.~\ref{section:process} details the data preprocessing pipeline for WiFi \ac{csi}. Finally, Sec.~\ref{sec:environmental information} introduces the formulation of geometric environmental descriptors derived from 3D point clouds, which serves as a critical geometry-aware prior for the proposed model.

\subsection{System Scenario and Channel Modeling}~\label{section:scenario}

\begin{figure}[!t]
	\centering
	\includegraphics[width=1\columnwidth]{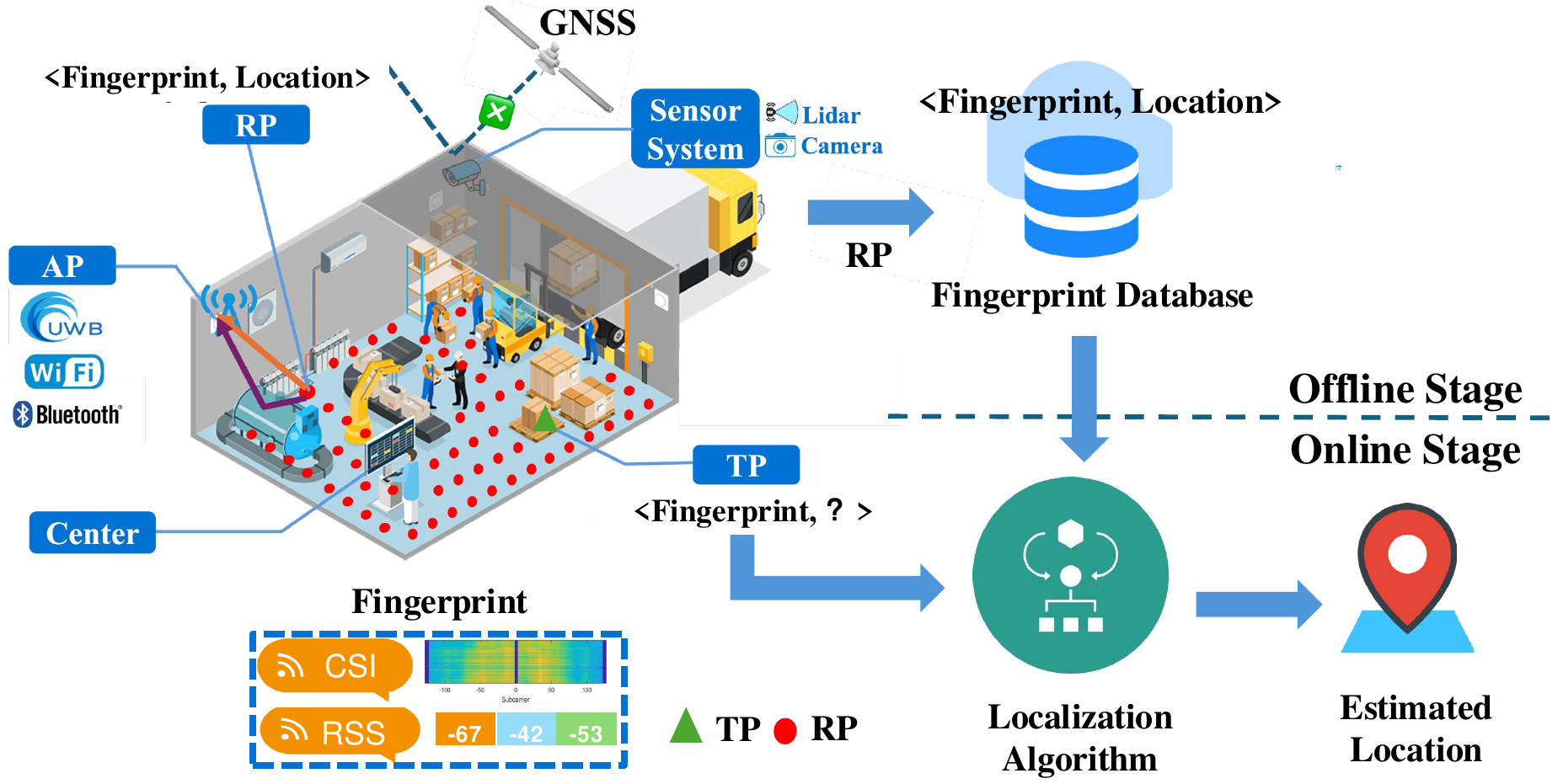}
	\caption{The scenario of fingerprinting-based localization in an uplink wireless communication system.}
	\label{fig:area}
\end{figure}
As illustrated in Fig.~\ref{fig:area}, we consider an uplink communication system consisting of $Z$ fixed WiFi \acpl{ap} and a user equipped with a mobile phone. Both the \acpl{ap} and the mobile device are equipped with omnidirectional antennas. Specifically, we define a survey area of size $A=(v \times w)\,\mathrm{m}^2$, where $v$ and $w$ represent the length and width, respectively. The area is discretized into $G=\left\lceil\frac{v}{\Delta r}\right\rceil \times \left\lceil\frac{w}{\Delta r}\right\rceil$ grid points with a uniform spacing of $\Delta r\,\mathrm{m}$, where $\lceil \cdot \rceil$ denotes the ceiling function. The set of reference grid-point locations in a 3D Cartesian coordinate system is denoted as $\left\{\boldsymbol{p}_g=\left(x_g, y_g, h_g\right) \mid g=1, \ldots, G\right\}$.

The wireless propagation channel between the user and the \acpl{ap} comprises a direct \ac{los} path and multiple indirect \ac{nlos} paths resulting from reflections, diffractions, and scattering. Frequency-domain \ac{csi} effectively captures these small-scale multipath effects. Specifically, the \ac{csi} measurement at the $k$-th subcarrier frequency $f_k$ is represented as a complex number:
\begin{equation} 
	H\left(f_k\right)=\left|H\left(f_k\right)\right| e^{j \angle H(f_k)}, \quad k=1, \ldots, K, \label{eq:csi} 
\end{equation}
where $\left|H\left(f_k\right)\right|$ denotes the amplitude and $\angle H(f_k)$ represents the phase. The total number of subcarriers $K$ depends on the configured channel bandwidth, as specified by the IEEE 802.11ac standard~\cite{ieee80211ac}: 64 for 20 MHz, 128 for 40 MHz, and 256 for 80 MHz.

\subsection{CSI Fingerprint Preprocessing}~\label{section:process}
WiFi fingerprinting is attractive for indoor localization because it can be implemented on commodity devices to provide fine-grained \ac{csi} measurements without dedicated ranging hardware~\cite{gringoli2019free}. In this work, we use \ac{csi} amplitude as the primary localization fingerprint because raw \ac{csi} phases on commercial WiFi devices are vulnerable to packet-level random offsets caused by hardware imperfections, such as carrier frequency offsets and sampling frequency offsets. Following the preprocessing practice in CRISLoc~\cite{CRISLoc}, we first zero out guard subcarriers that may contain random or meaningless values. We then apply Mahalanobis-distance-based packet screening to the CSI amplitude matrix and discard the top 10\% packets with the largest deviations from the overall distribution. This procedure removes bursty abnormal packets while preserving the dominant spatial pattern in the amplitude fingerprints.

\begin{figure}[!t]
	\centering
	\subfigure{\includegraphics[scale=0.35]{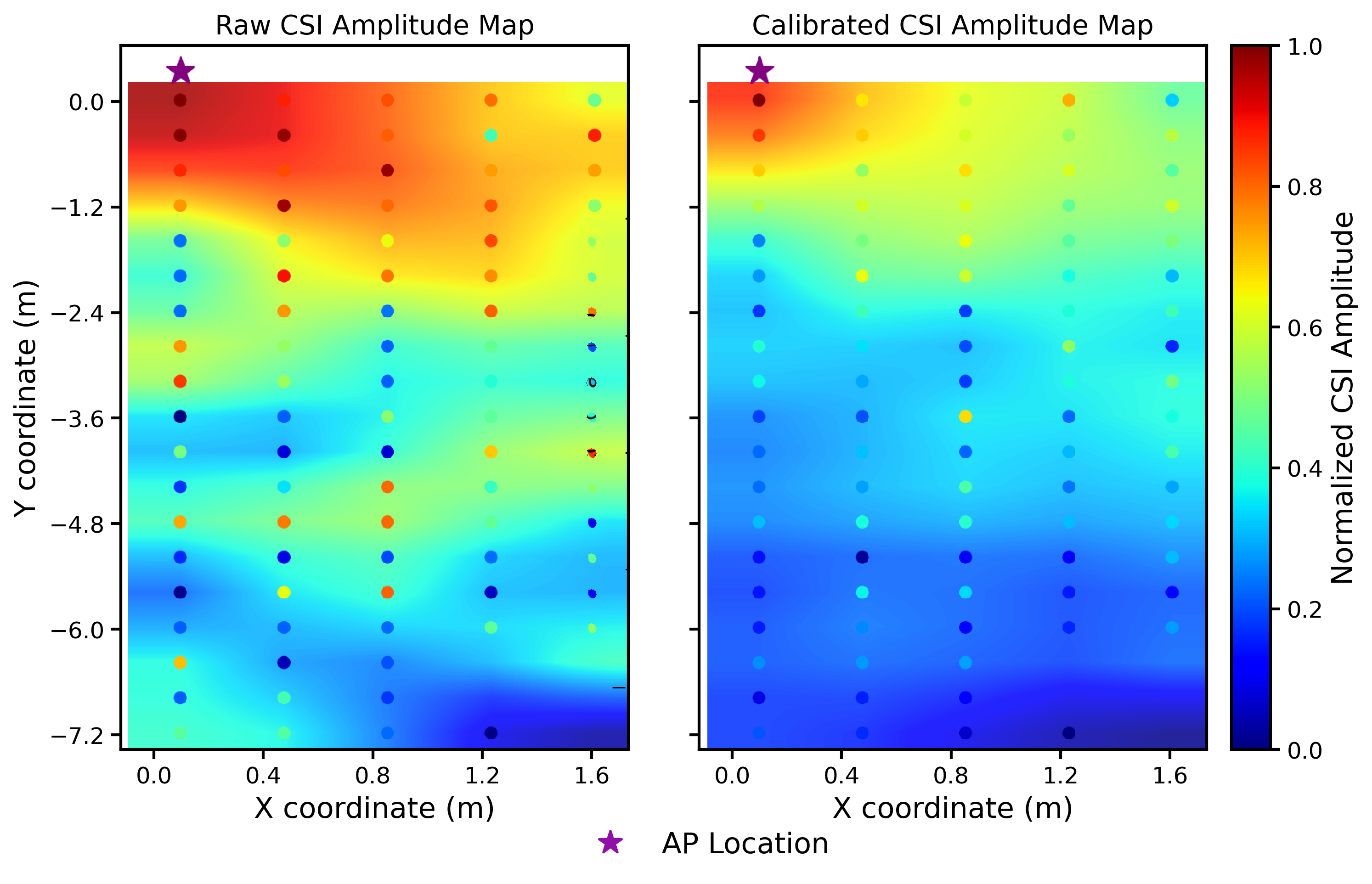}}
	\caption{Raw and RSSI-calibrated CSI amplitude heatmaps for the same \ac{ap}. The purple star marks the \ac{ap} position at (0, 0.5).}
	\label{fig:amplitude}
\end{figure}

The raw CSI amplitude reported by commercial WiFi NICs is also heavily scaled by the \ac{agc} mechanism, which weakens its consistency with the physical path-loss trend~\cite{CRISLoc}. To mitigate this effect, we calibrate the CSI amplitude using the corresponding \ac{rssi} recorded in each packet header. Since \ac{agc} mainly introduces a scalar gain across the frequency band, a uniform calibration factor is applied to all subcarriers within the same packet. Fig.~\ref{fig:amplitude} compares the raw and calibrated CSI amplitude maps. Before calibration, survey points closer to the \ac{ap} may still show lower amplitudes. After RSSI-based calibration, the amplitude distribution becomes more consistent with the path-loss trend, with larger values appearing closer to the \ac{ap}. This calibrated amplitude is therefore more suitable for spatial localization.

\subsection{Environmental Descriptor Extraction}\label{sec:environmental information}
To characterize propagation-relevant environmental factors, we incorporate geometric information from the surrounding environment into the localization framework. The environmental data are collected using a Leica RTC360 LT terrestrial laser scanner, which measures reflected laser signals and generates a 3D point cloud of the survey area~\cite{leica_rtc360lt_datasheet}. We represent the point cloud as~\cite{10757755}
\begin{equation}
	P_{\text{env}} = \{(p_i, s_i) \mid i = 1, \dots, I\},
\end{equation}
where $p_i$ denotes the 3D coordinate of the $i$-th point, $s_i$ denotes its auxiliary attribute such as color or material label, and $I$ is the number of points. To align the point cloud with the WiFi survey grid, we register it to the global WiFi coordinate frame using at least three markers placed at known grid locations. The corresponding 3D rigid transformation is estimated by least-squares fitting in CloudCompare.

\begin{figure}[!t]
	\centering
	\includegraphics[scale=0.3]{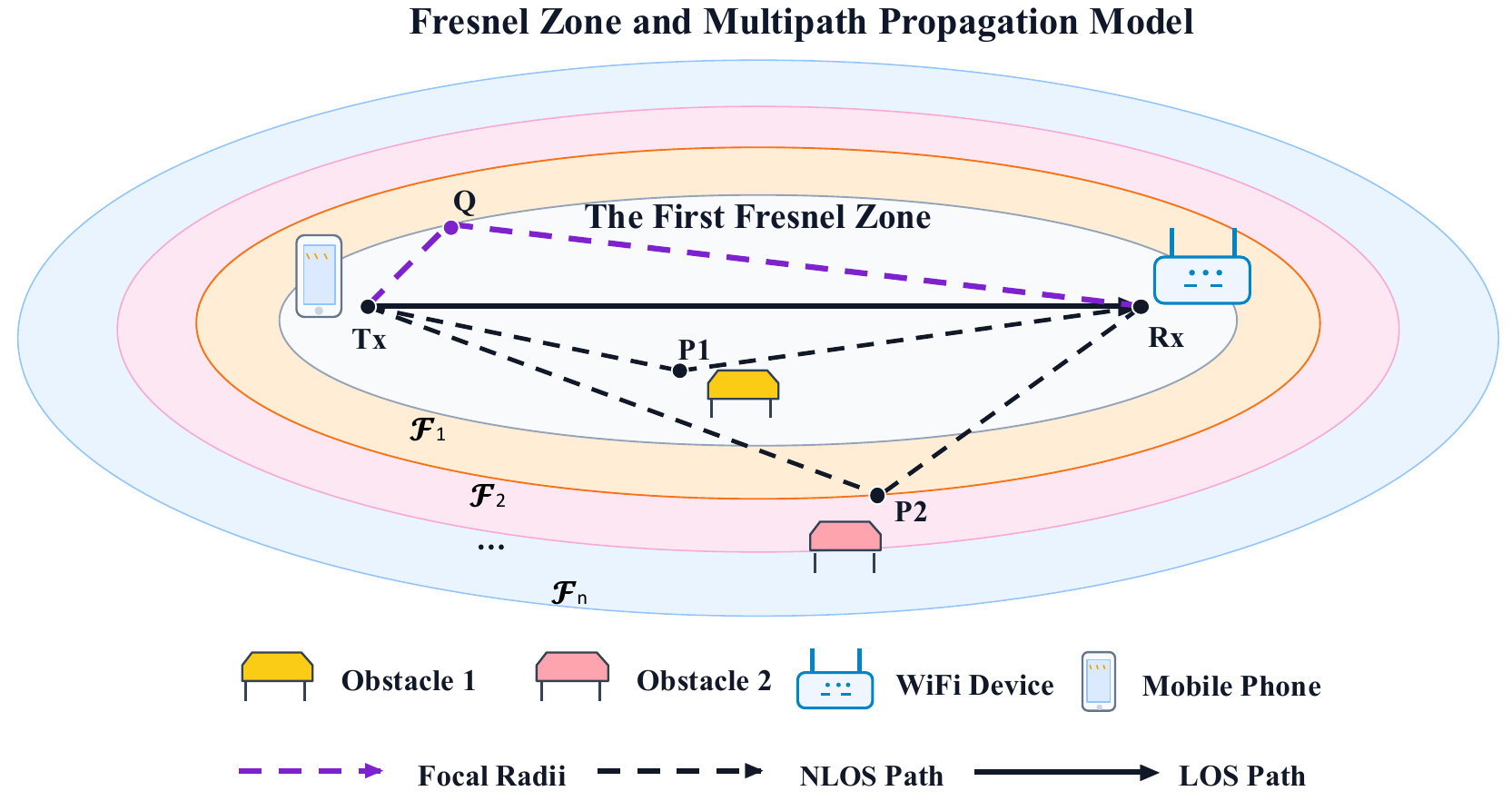}
	\caption{Illustration of the Fresnel-zone geometry in an uplink wireless communication system.}
	\label{fig:tr}
\end{figure}

To connect the point cloud with wireless propagation, we use the Fresnel zone model~\cite{wu2022wifi}. As illustrated in Fig.~\ref{fig:tr}, the Fresnel zones $\mathcal{F}_n \ (n=1,2,\dots,N)$ are a set of confocal ellipsoids whose two foci are the transmitter (Tx) and receiver (Rx). A point $Q_n$ on the boundary of the $n$-th Fresnel zone satisfies
\begin{equation}
	|T_x Q_n| + |Q_n R_x| - |T_x R_x| = n \cdot \frac{\lambda_c}{2},\quad n = 1,\dots,N,
\end{equation}
where $\lambda_c$ is the carrier wavelength. The first Fresnel zone corresponds to the innermost ellipsoid with $n=1$. For WiFi devices operating at 5~GHz, where $\lambda_c \approx 6$~cm, this zone forms a narrow region around the \ac{los} path. Obstacles inside this region can change the received signal through reflection, diffraction, and absorption. Therefore, we extract environmental information from the first Fresnel zone, which contains the main propagation region of the direct path~\cite{wu2022wifi,yang2024orchloc}.

Based on the registered point cloud and the first Fresnel zone, we construct an environmental descriptor $\mathbf{y}$ for each survey location. For each \ac{ap}-user link, we identify the first Fresnel zone and collect the point-cloud points inside it. These points are grouped into $C$ material categories, and the proportion of each material is calculated. The resulting material-proportion vector summarizes the local propagation environment, since different materials affect radio signals in different ways; for example, metal tends to cause strong reflection, while glass may allow partial transmission. We concatenate the material-proportion vectors from all $Z$ \acpl{ap} to obtain the final environment-conditioned descriptor $\mathbf{y} \in \mathbb{R}^{Z \cdot C}$ for each survey location.

\begin{figure}[!t]
	\centering
	\subfigure[]{\includegraphics[scale=0.4]{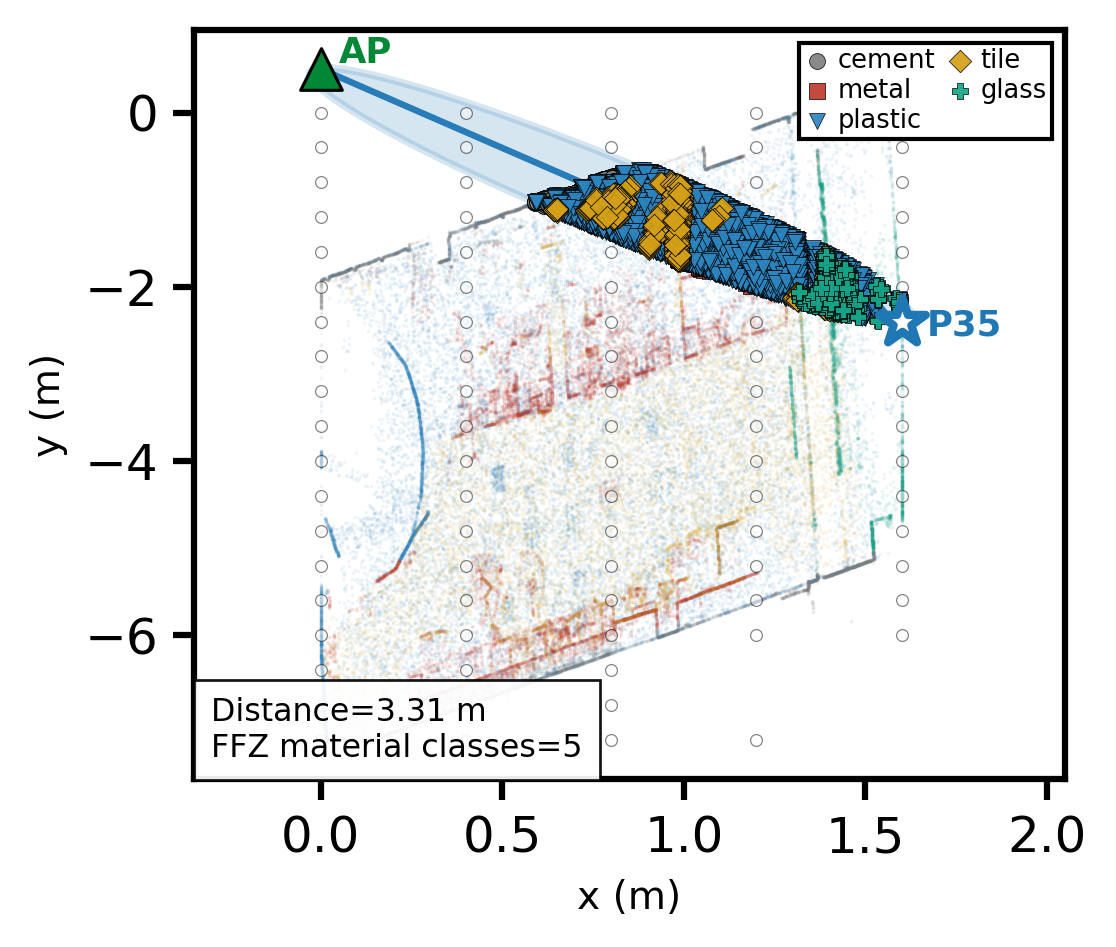}}
	\subfigure[]{\includegraphics[scale=0.4]{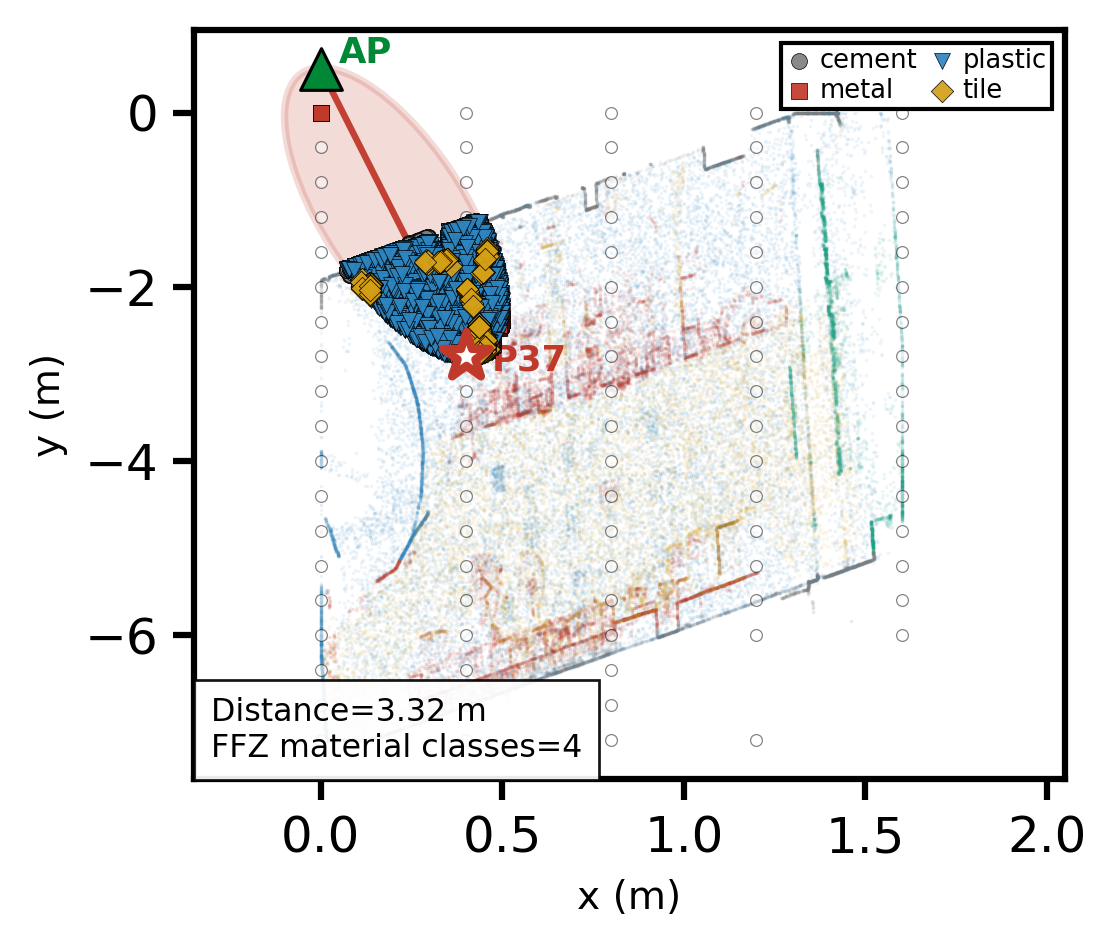}}
	\subfigure[]{\includegraphics[scale=0.4]{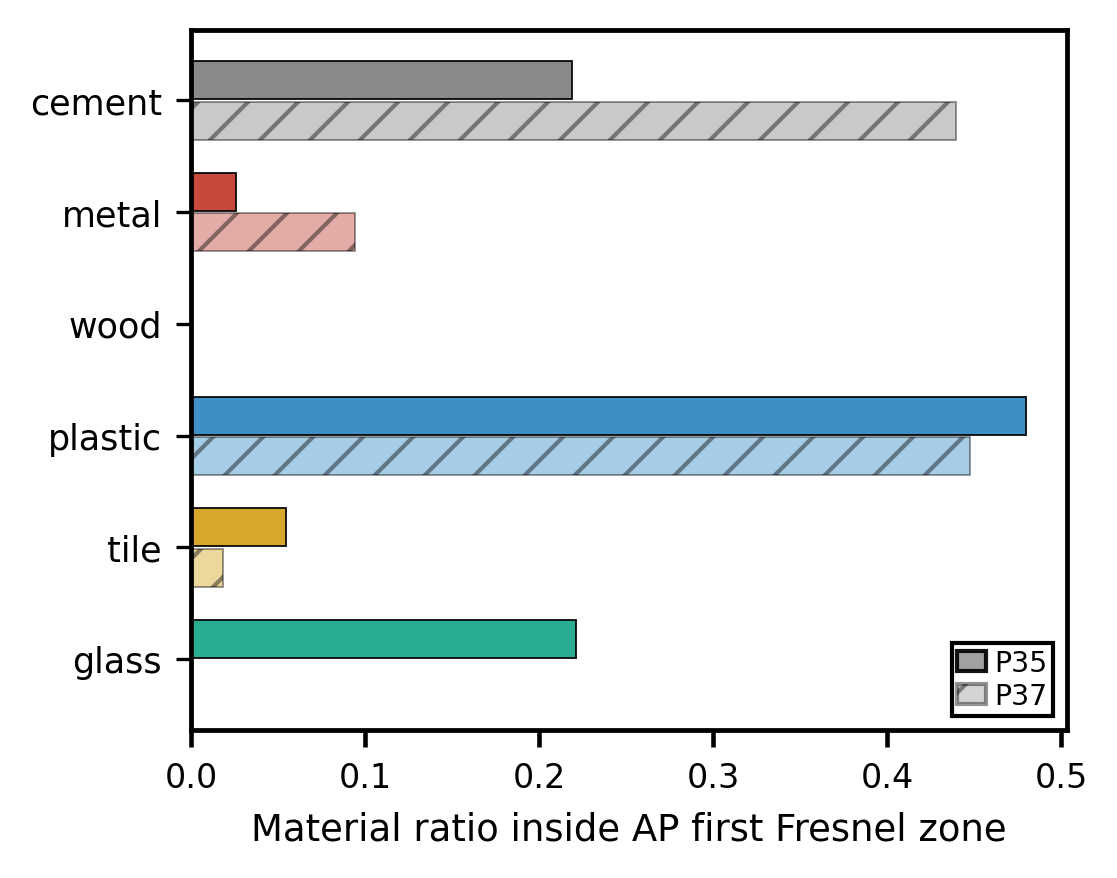}}
	\subfigure[]{\includegraphics[scale=0.4]{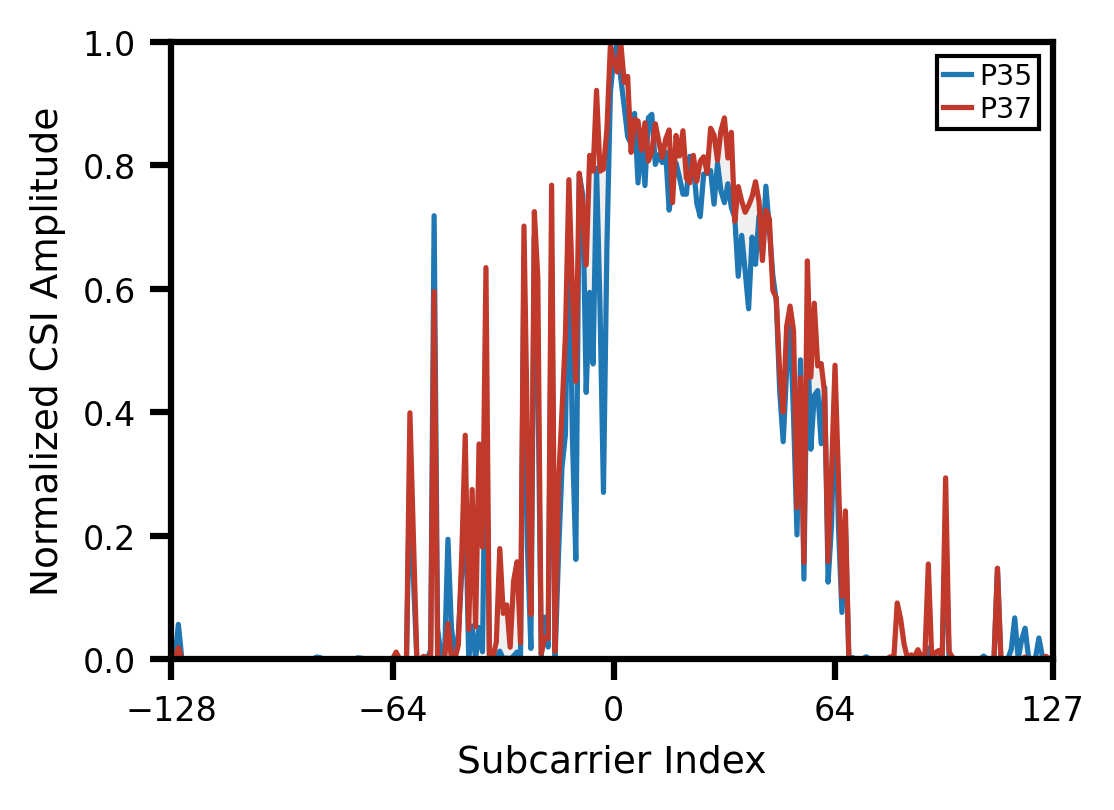}}
	\caption{An example of CSI ambiguity between two receiver locations. Positions P35 and P37 are separated by 1.26~m but have nearly identical AP distances and similar CSI amplitude responses. Their first Fresnel zone obstacle distributions and material compositions are different, showing that material-aware environmental information can help distinguish CSI-ambiguous locations.}
	\label{fig:ambiguity}
\end{figure}

The CSI ambiguity case in Fig.~\ref{fig:ambiguity} further motivates the use of environmental descriptors. Positions P35 and P37 are separated by 1.26~m, while their distances to the same AP are nearly identical, at 3.312~m and 3.324~m, respectively. Fig.~\ref{fig:ambiguity}(d) shows that the two links also have similar CSI amplitude responses. Thus, CSI amplitude alone may not provide enough information to distinguish spatially separated locations with similar AP distances.

In contrast, the corresponding first Fresnel zones contain different obstacle distributions and material compositions. Figs.~\ref{fig:ambiguity}(a) and (b) show the first Fresnel zone regions of the two links, where point-cloud obstacles are colored by material category. Fig.~\ref{fig:ambiguity}(c) compares their material proportions. These differences suggest different reflection, diffraction, and attenuation conditions along the two propagation paths. Therefore, the material-aware point-cloud descriptor provides additional environmental information for resolving CSI ambiguity in localization.

\section{Environment-Conditioned Meta-Learning}
\label{sec:ECML}

Meta-learning optimizes a shared set of meta-parameters $\theta$ as a task-agnostic initialization, enabling rapid adaptation to new tasks with a few gradient steps. In the considered localization setting, the network $f_\theta(\cdot)$ maps CSI measurements to 3D location vectors $\boldsymbol{p}=(x,y,h)$. Each task $\tau_i \sim \mathcal{P}(\tau)$ corresponds to a randomly sampled sub-region within a scenario, which is treated as a local environment where obstacle layouts and multipath conditions are approximately shared. Each task comprises a support set $D_{\tau_i}^s$ of $N_s$ labeled CSI--location pairs for task-specific gradient adaptation, and a query set $D_{\tau_i}^q$ of $N_q$ labeled pairs for evaluating the adapted model. In addition, the support samples are associated with the point-cloud-derived environmental descriptors defined in Sec.~\ref{sec:environmental information}; these descriptors are used to condition the meta-initialization via the diffusion generator.

The support-set environmental descriptors provide explicit information about the propagation environment, such as geometry and material properties. We organize them into a task-level matrix $\mathbf{Y}_{\tau_i} = [\mathbf{y}_{\tau_i,1}, \mathbf{y}_{\tau_i,2}, \dots, \mathbf{y}_{\tau_i,N_s}]^\top$,
where $\mathbf{y}_{\tau_i,j} \in \mathbb{R}^{Z \cdot C}$ is the descriptor of the $j$-th support point, so that $\mathbf{Y}_{\tau_i} \in \mathbb{R}^{N_s \times ZC}$. Under environment-dependent distribution shifts, a single shared initialization $\theta$ may lie far from the task-specific optimum, particularly when the support set is small. By leveraging $\mathbf{Y}_{\tau_i}$ to capture propagation-relevant geometric cues, we learn an environment-conditioned offset $\Delta\theta_{\tau_i}$ and construct the modulated initialization $\theta_i^{(0)} = \theta + \Delta\theta_{\tau_i}$.

The additive form is motivated by two considerations. First, following the residual-learning principle~\cite{he2016deep}, the generator produces a small environment-conditioned correction $\Delta\theta_{\tau_i}$ that adjusts the shared meta-initialization~$\theta$ toward the target environment before gradient-based fine-tuning on wireless support samples. This is considerably easier than generating the full task-specific parameter vector from scratch. Second, the offset $\Delta\theta_{\tau_i}$ injects geometric propagation knowledge extracted from point-cloud descriptors, complementing the wireless-signal-based meta-initialization with environment-specific information that is difficult to infer from a few CSI measurements alone.

\subsection{Latent Diffusion Generator for Meta-Parameter Offsets}
\label{sec:diffusion_gen}

Based on the above task representation, we introduce an environment-conditioned latent diffusion generator that models the conditional distribution $p_\phi(\Delta\theta \mid \mathbf{Y})$, where $\mathbf{Y}$ denotes the support-set environmental descriptor matrix and $\Delta\theta$ represents task-specific parameter offsets. The diffusion state $z_t$ is represented as a one-dimensional parameter-offset sequence whose length equals $\mathrm{dim}(\theta)$. The latent dimension used below refers to the hidden channel width of the U-Net denoiser rather than to the dimensionality of the generated offset vector. We adopt this conditional diffusion formulation because similar environmental descriptors may still require different parameter corrections under local multipath variations, making a conditional distribution more appropriate than a deterministic mapping. The generator therefore consists of a forward noising process and a reverse denoising process.

\begin{figure}
    \centering
    \includegraphics[scale=0.32]{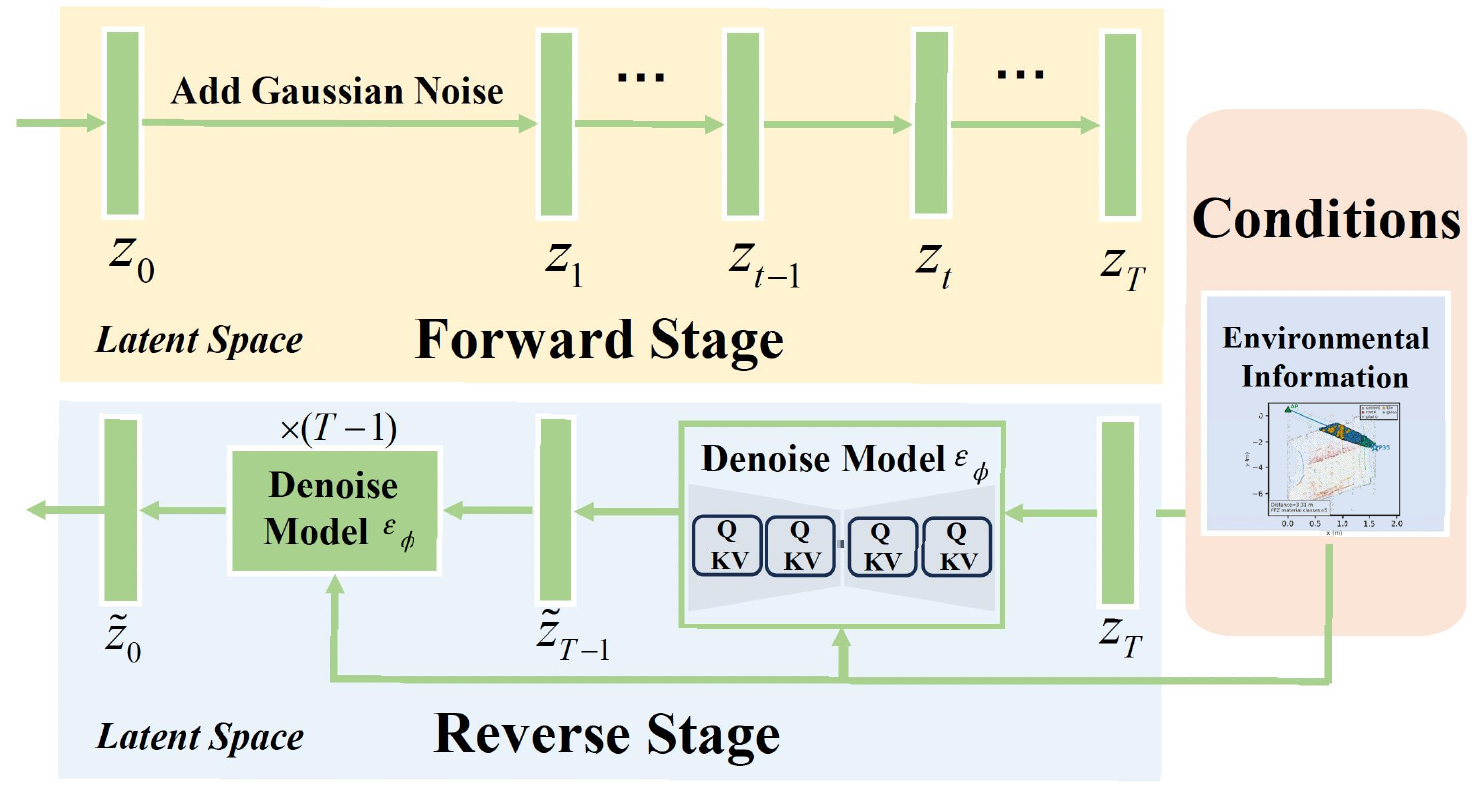}
    \caption{Workflow of the environment-conditioned denoising model.}
    \label{fig:denoise}
\end{figure}

\begin{itemize}
    \item \textbf{Forward process:} 
    As illustrated in Fig.~\ref{fig:denoise}, let $z_0$ denote the task-specific parameter-offset sequence corresponding to $\Delta\theta$. The forward process progressively corrupts $z_0$ over $T_{\mathrm{d}}$ steps with a noise schedule $\{\alpha_t\}_{t=1}^{T_{\mathrm{d}}}$, $\alpha_t \in (0,1)$, equivalently $\alpha_t=1-\beta_t^{\mathrm{d}}$:
    \begin{equation}
    z_t = \sqrt{\alpha_t}\, z_{t-1} + \sqrt{1-\alpha_t}\, \epsilon_t, \quad \epsilon_t \sim \mathcal{N}(0, \mathbf{I}).
    \end{equation}
    Letting $\bar{\alpha}_t = \prod_{s=1}^{t} \alpha_s$, the marginal admits the closed form $z_t = \sqrt{\bar{\alpha}_t}\, z_0 + \sqrt{1-\bar{\alpha}_t}\, \epsilon$,
    and the schedule ensures $\bar{\alpha}_{T_{\mathrm{d}}} \approx 0$ so that $z_{T_{\mathrm{d}}} \approx \mathcal{N}(0, \mathbf{I})$.

    \item \textbf{Reverse process:} 
    The reverse diffusion process progressively recovers $z_0$ from $z_{T_{\mathrm{d}}} \sim \mathcal{N}(0,\mathbf{I})$.
    It is parameterized by a denoising network 
    $\epsilon_\phi(z_t, t, \mathbf{Y})$, which predicts the noise component $\epsilon$ at step $t$ 
    conditioned on the support-set descriptor matrix $\mathbf{Y}$.
    At each step, the denoised estimate is obtained via
    \begin{equation}
    \tilde{z}_{t-1} = \frac{1}{\sqrt{\alpha_t}}\!\left(z_t - \frac{1-\alpha_t}{\sqrt{1-\bar{\alpha}_t}}\, \epsilon_\phi(z_t, t, \mathbf{Y})\right) + \sigma_t\, \mathbf{n},
    \end{equation}
    where $\mathbf{n} \sim \mathcal{N}(0,\mathbf{I})$ for $t > 1$, with $\mathbf{n}=0$ for $t=1$, and $\sigma_t$ is the posterior standard deviation determined by the noise schedule.
    The denoising network adopts a U-Net backbone augmented with cross-attention layers 
    to inject environmental context into the denoising dynamics.
    Specifically, the cross-attention mechanism is defined as
    \begin{equation}
    \text{Attention}(\mathbf{Q},\mathbf{K},\mathbf{V}) = 
    \text{softmax}\!\left(\frac{\mathbf{Q}\mathbf{K}^\top}{\sqrt{d}}\right) \mathbf{V},
    \end{equation}
    where $\mathbf{Q} = W_Q \varphi(z_t)$, $\mathbf{K} = W_K \mathbf{Y}$, and $\mathbf{V} = W_V \mathbf{Y}$.
    Here, $\varphi(z_t)$ denotes a flattened intermediate U-Net feature map and does not introduce additional learnable parameters,
    while $W_Q$, $W_K$, and $W_V$ are learnable projection matrices and $d$ is the feature dimension.
    This design enables the denoising network to adapt the reverse diffusion process to environment-specific characteristics.
\end{itemize}

The generated offset $\Delta\theta$ has the same dimensionality as the localization-network parameter vector~$\theta$, providing environment-specific corrections to all weights and biases.
For the localization network used in the experiments, the input dimension is $Z\times256=768$ with $Z=3$ APs and the hidden dimensions are $[256,128,64]$. This gives $\mathrm{dim}(\theta)\approx 238\mathrm{K}$ trainable weights and biases. The U-Net processes the offset sequence with shared convolutional kernels and a latent channel width of 64. Its final $1\times 1$ convolution maps the hidden channels back to a one-channel offset sequence of length $\mathrm{dim}(\theta)$, avoiding a dense projection from a 64-dimensional code to the full parameter vector.

\begin{figure}
    \centering
	\includegraphics[width=\columnwidth]{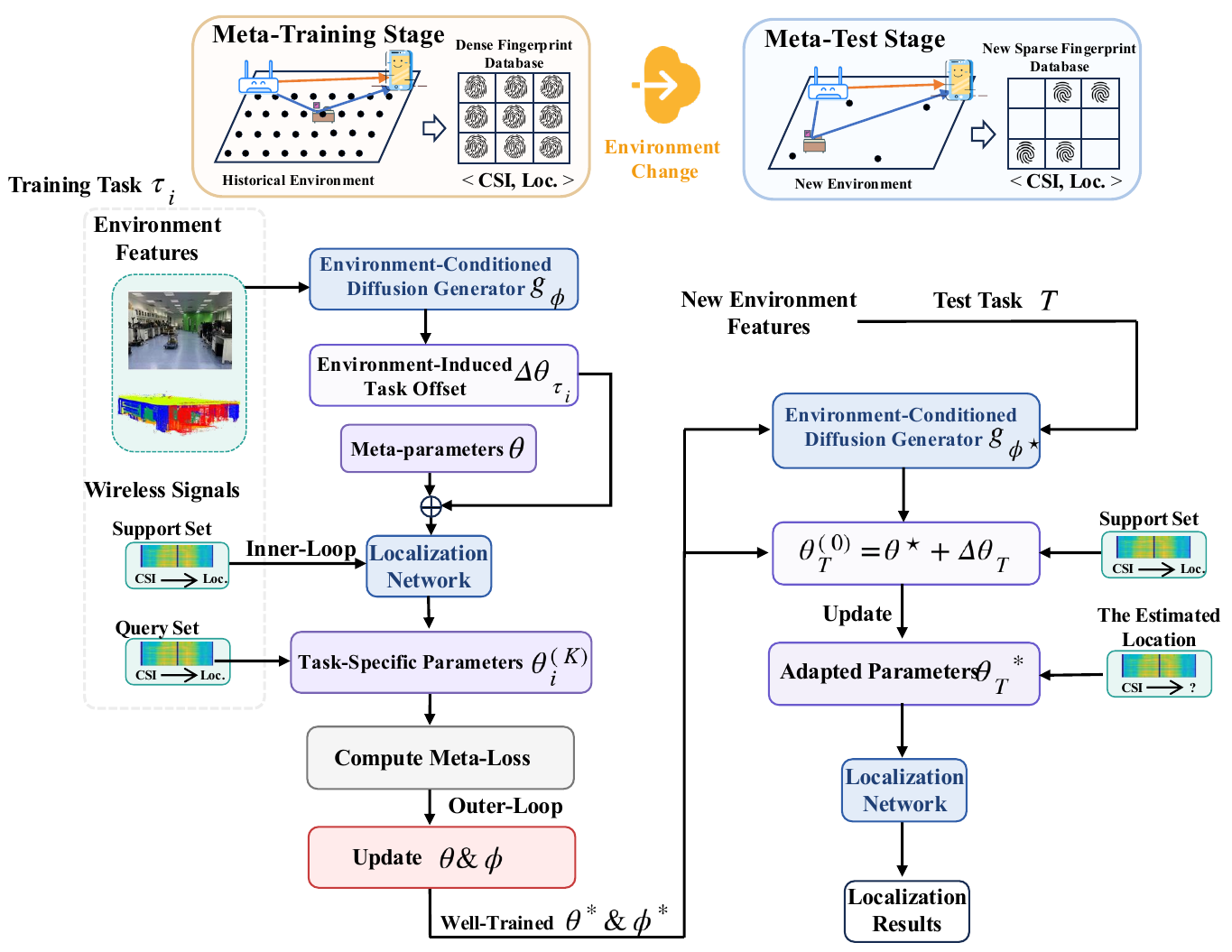}
    \caption{Environment-conditioned meta-learning framework for wireless fingerprint localization.}
    \label{fig:envcond-framework}
\end{figure}
\subsection{Environment-Conditioned Meta-Training}

During the meta-training stage, the meta-parameters \(\theta\) and diffusion generator parameters \(\phi\) are jointly optimized over a set of training tasks
\(\{\tau_i\}_{i=1}^M \sim \mathcal{P}(\tau)\).
For each task, the diffusion generator produces a task-specific parameter offset:
\begin{equation}
\theta_i^{(0)} = \theta + g_\phi(z_{\tau_i}, \mathbf{Y}_{\tau_i}),
\end{equation}
where \(z_{\tau_i} \sim \mathcal{N}(0,\mathbf{I})\) is a latent variable sampled from the standard Gaussian prior and transformed by the reverse diffusion process conditioned on \(\mathbf{Y}_{\tau_i}\).
This initialization captures environment-dependent variations that are difficult to learn from few support samples alone, effectively shaping the prior distribution of task parameters.

The initialized parameters are then refined through \(K_{\mathrm{in}}\) steps of gradient-based adaptation on the support set:
\begin{equation}
\theta_i^{(k+1)} = \theta_i^{(k)} - \alpha \nabla_{\theta_i^{(k)}} \mathcal{L}_{\tau_i}(\theta_i^{(k)}; D_{\tau_i}^s), 
\quad k = 0,\dots,K_{\mathrm{in}}-1,
\end{equation}
where \(\alpha\) is the inner-loop step size. 
The performance of the adapted parameters is evaluated on the query set
\(D_{\tau_i}^q\), and the meta-loss is computed as the sum of query losses
across all tasks:
\begin{equation}
\mathcal{L}_{\mathrm{meta}}(\theta, \phi)
= \sum_{i=1}^{M}
\mathcal{L}_{\tau_i}\!\left(\theta_i^{(K_{\mathrm{in}})}; D_{\tau_i}^q\right).
\end{equation}

The meta-parameters and diffusion generator parameters are learned by solving
\begin{equation}
(\theta^\star, \phi^\star)
= \arg\min_{\theta,\,\phi}
\;\mathbb{E}_{\{\tau_i\} \sim \mathcal{P}(\tau)}
\big[ \mathcal{L}_{\mathrm{meta}}(\theta, \phi) \big],
\label{eq:meta_obj}
\end{equation}
which is optimized in practice via gradient-based outer-loop updates:
\begin{equation}
(\theta, \phi) \leftarrow (\theta, \phi)
- \beta \nabla_{(\theta,\phi)} \mathcal{L}_{\mathrm{meta}}(\theta, \phi),
\end{equation}
where \(\beta\) is the outer-loop step size. Here, $\theta^\star$ encodes shared wireless-signal representations learned from CSI fingerprints across environments, while $\phi^\star$ captures how physical environment geometry, as described by point-cloud descriptors, modulates the task-specific parameter corrections.

Two notable properties emerge from the proposed joint optimization framework.
First, since ground-truth offsets are unavailable, the parameterized generator $\phi$ cannot be trained using a standard noise-prediction objective. Instead, the diffusion generator is optimized end-to-end via the meta-loss: the generated offset is injected into the localization network, and the resulting query-set loss after inner-loop adaptation is back-propagated through both the gradient-based updates and the reverse diffusion process. Consequently, $\mathcal{L}_{\mathrm{meta}}$ in Eq.~\eqref{eq:meta_obj} serves as the sole supervision signal for $\phi$. Second, no explicit norm regularization is imposed on $\Delta\theta$. Instead, its magnitude is implicitly controlled by the meta-objective. If the diffusion generator produces an offset that moves the initialization $\theta + \Delta\theta$ excessively far from the task-specific optimum, the limited $K_{\mathrm{in}}$-step adaptation on a small support set may be insufficient to correct this mismatch. As a result, the query-set performance decreases and the meta-loss $\mathcal{L}_{\mathrm{meta}}$ increases. The resulting outer-loop gradients then penalize such harmful offsets through the optimization of $\phi$, thereby discouraging parameter shifts that are too large to be absorbed by the inner-loop adaptation. Therefore, although no auxiliary $\ell_2$ penalty is explicitly imposed on $\Delta\theta$, the bilevel objective still provides an implicit control over its magnitude.

\subsection{Environment-Conditioned Meta-Test}

In the meta-test stage, the diffusion generator parameters $\phi^\star$ obtained during meta-training are frozen, since the generator has already learned a generalizable mapping from environmental descriptors to parameter offsets across diverse training tasks, and fine-tuning it on a single sparse support set would risk overfitting. Given a new environment with sparse support set $D_T^s$ and descriptor matrix $\mathbf{Y}_T$, the frozen generator produces an environment-conditioned offset $\Delta\theta_T = g_{\phi^\star}(z, \mathbf{Y}_T)$, and the task-specific parameters are initialized as
\begin{equation}
\theta_T^{(0)} = \theta^\star + g_{\phi^\star}(z, \mathbf{Y}_T), 
\quad z \sim \mathcal{N}(0, \mathbf{I}).
\end{equation}

The initialization is then refined via $K_{\mathrm{in}}$ steps of gradient-based adaptation on the support set, producing the adapted parameters:
\begin{equation}
\theta_T^{(k+1)} = \theta_T^{(k)} - \alpha \nabla_{\theta_T^{(k)}} \mathcal{L}_T(\theta_T^{(k)}; D_T^s), 
\quad k = 0, \dots, K_{\mathrm{in}}-1.
\end{equation}
This two-stage procedure decouples the roles of the diffusion generator and the inner-loop optimizer: the generator injects geometric prior knowledge from $\mathbf{Y}_T$ to bring the initialization close to the target region, while gradient adaptation compensates for the residual task-specific discrepancy using the wireless support samples.

We further define the task-specific empirical optimum as
\begin{equation}
\theta_T^\star = \arg\min_\theta \mathcal{L}_T(\theta; D_T^s).
\end{equation}
Due to the limited adaptation budget of $K_{\mathrm{in}}$ steps and the sparsity of the support data, the adapted parameters $\theta_T^{(K_{\mathrm{in}})}$ serve as a practical approximation to $\theta_T^\star$. In other words, $\theta_T^{(K_{\mathrm{in}})}$ reflects the realized outcome of finite-step adaptation, whereas $\theta_T^\star$ represents the ideal empirical optimum. This distinction provides a principled basis for analyzing bias, variance, and convergence in unseen environments, as presented in the next section.

\section{Theoretical Analysis}
\label{sec:theory}

This section analyzes EnvCoLoc under local regularity conditions for finite-step adaptation. The analysis serves three purposes. First, it relates the adaptation loss to the distance between the generated initialization and the task-specific optimum. Second, it identifies when environmental descriptors reduce the expected initialization error relative to an environment-agnostic initializer. Third, it separates the contribution of stochastic diffusion generation from that of a deterministic offset mapping.

\subsection{Preliminaries and Assumptions}
\label{sec:assumptions}

Let \(p_{\tau_i}\) denote the data distribution of task \(\tau_i\). For a localization model with parameter \(\theta\), let \({\cal L}_{\tau_i}(\theta)\) denote the task loss used for finite-step adaptation, implemented as the MSE over the support samples of task \(\tau_i\). Let \(\theta_{\tau_i}^\star\) be the corresponding optimal task-specific parameter. Given the shared meta-initialization \(\theta^\star\), define the ideal task offset as
\begin{equation}
\Delta_{\tau_i}^\star := \theta_{\tau_i}^\star-\theta^\star .
\end{equation}
The environment-conditioned initialization of task \(\tau_i\) is
\begin{equation}
\theta_i^{(0)}=\theta^\star+g_{\phi}(z_{\tau_i},\mathbf{Y}_{\tau_i}),
\end{equation}
where \(g_{\phi}\) is the stochastic generator, \(z_{\tau_i}\) is the latent variable, and \(\mathbf{Y}_{\tau_i}\) is the environmental descriptor. The following assumptions are used for the local analysis.

\begin{assumption}[Local Optimization Regularity]
\label{assump:basic}
There exists a neighborhood of \(\theta_{\tau_i}^\star\) that contains the inner-loop iterates \(\{\theta_i^{(k)}\}_{k=0}^{K_{\mathrm{in}}}\). Within this neighborhood, each task loss is \(L\)-smooth and \(\mu\)-strongly convex, with a local Hessian upper bound \(\nabla^2 {\cal L}_{\tau_i} \preceq \lambda I\). The inner-loop step size satisfies \(0<\alpha\le 1/L\).
\end{assumption}

\begin{assumption}[Generator Statistical Bounds]
\label{assump:generator_bound}
For each task \(\tau_i\), the generated offset admits the decomposition
\begin{equation}
  g_{\phi}(z_{\tau_i},\mathbf{Y}_{\tau_i})
  =m_{\tau_i}+\xi_{\tau_i}+\rho_{\tau_i},
\end{equation}
where \(m_{\tau_i}=\mathbb{E}[g_{\phi}(z_{\tau_i},\mathbf{Y}_{\tau_i})|\mathbf{Y}_{\tau_i}]\) is the deterministic component, \(\xi_{\tau_i}\) denotes the stochastic sampling fluctuation, and \(\rho_{\tau_i}\) denotes the residual generation error. With expectations over tasks and generator randomness,
\begin{equation}
\mathbb{E}\|m_{\tau_i}-\Delta_{\tau_i}^\star\|_2^2\le\psi_b,
\quad
\mathbb{E}\|\xi_{\tau_i}\|_2^2\le\psi_v,
\quad
\mathbb{E}\|\rho_{\tau_i}\|_2^2\le\psi.
\end{equation}
The terms \(\xi_{\tau_i}\) and \(\rho_{\tau_i}\) are zero-mean, and the cross terms vanish in expectation.
\end{assumption}

\begin{assumption}[Score and Sampling Regularity]
\label{assump:score}
The effective score, equivalently noise-prediction, estimation error of the denoising network is bounded by \(\epsilon_{\mathrm{score}}^2\). The reverse-time sampler satisfies standard stability and discretization regularity conditions, so the total sampling residual can be decomposed into an estimation-error component and a numerical discretization component.
\end{assumption}

\subsection{Convergence and Excess-Loss Bounds}
\label{sec:main_results}

\begin{lemma}[Inner-Loop Contraction]
\label{lem:adapt}
Under Assumption~\ref{assump:basic}, after \(K_{\mathrm{in}}\) inner-loop gradient steps initialized at \(\theta_i^{(0)}\), the adapted parameter satisfies
\begin{equation}
\label{eq:lemma1}
\left\|\theta_i^{(K_{\mathrm{in}})}-\theta_{\tau_i}^\star\right\|_2^2
\le
r^{K_{\mathrm{in}}}
\left\|\theta_i^{(0)}-\theta_{\tau_i}^\star\right\|_2^2,
\quad
r=(1-\alpha\mu)^2.
\end{equation}
\end{lemma}

\begin{proof}
For task \(\tau_i\), the inner-loop update is
\begin{equation}
\theta_i^{(k+1)}
=
\theta_i^{(k)}-
\alpha\nabla {\cal L}_{\tau_i}(\theta_i^{(k)}).
\end{equation}
Since \(\theta_{\tau_i}^\star\) minimizes \({\cal L}_{\tau_i}\), \(\nabla {\cal L}_{\tau_i}(\theta_{\tau_i}^\star)=0\). Subtracting \(\theta_{\tau_i}^\star\) from both sides gives
\begin{align}
\theta_i^{(k+1)}-\theta_{\tau_i}^\star
&=
\theta_i^{(k)}-\theta_{\tau_i}^\star 
\nonumber\\
&\quad-
\alpha\left(
\nabla {\cal L}_{\tau_i}(\theta_i^{(k)})-
\nabla {\cal L}_{\tau_i}(\theta_{\tau_i}^\star)
\right).
\end{align}
Under Assumption~\ref{assump:basic}, standard gradient-descent contraction for an \(L\)-smooth and \(\mu\)-strongly convex loss yields
\begin{equation}
\left\|\theta_i^{(k+1)}-\theta_{\tau_i}^\star\right\|_2^2
\le
(1-\alpha\mu)^2
\left\|\theta_i^{(k)}-\theta_{\tau_i}^\star\right\|_2^2.
\end{equation}
Applying the inequality recursively from \(k=0\) to \(K_{\mathrm{in}}-1\) gives \eqref{eq:lemma1}.
\end{proof}

\begin{theorem}[Conditional Local Excess-Loss Bound]
\label{thm:diffusion_meta}
Under Assumptions~\ref{assump:basic} and~\ref{assump:generator_bound}, define the task-averaged support-set excess loss as
\begin{equation}
\mathcal{E}_{\mathrm{EnvCoLoc}}
:=
\mathbb{E}_{\tau_i,z}
\left[
{\cal L}_{\tau_i}(\theta_i^{(K_{\mathrm{in}})})
-
{\cal L}_{\tau_i}(\theta_{\tau_i}^\star)
\right].
\end{equation}
Then
\begin{equation}
\label{eq:meta_diffusion_bound}
\mathcal{E}_{\mathrm{EnvCoLoc}}
\le
\overline{\mathcal{E}}_{\mathrm{EnvCoLoc}}
:=
\frac{\lambda}{2} r^{K_{\mathrm{in}}}
(\psi_b+\psi_v+\psi).
\end{equation}
\end{theorem}

\begin{proof}
By Taylor's theorem and the optimality condition \(\nabla {\cal L}_{\tau_i}(\theta_{\tau_i}^\star)=0\), the local Hessian upper bound in Assumption~\ref{assump:basic} implies
\begin{equation}
\label{eq:taylor_excess_bound}
{\cal L}_{\tau_i}(\theta_i^{(K_{\mathrm{in}})})
-
{\cal L}_{\tau_i}(\theta_{\tau_i}^\star)
\le
\frac{\lambda}{2}
\left\|\theta_i^{(K_{\mathrm{in}})}-\theta_{\tau_i}^\star\right\|_2^2.
\end{equation}
Using Lemma~\ref{lem:adapt},
\begin{equation}
\label{eq:inner_contraction_bound}
\left\|\theta_i^{(K_{\mathrm{in}})}-\theta_{\tau_i}^\star\right\|_2^2
\le
r^{K_{\mathrm{in}}}
\left\|\theta_i^{(0)}-\theta_{\tau_i}^\star\right\|_2^2.
\end{equation}
Combining \eqref{eq:taylor_excess_bound} and \eqref{eq:inner_contraction_bound} gives
\begin{equation}
\label{eq:excess_risk_init_error}
{\cal L}_{\tau_i}(\theta_i^{(K_{\mathrm{in}})})
-
{\cal L}_{\tau_i}(\theta_{\tau_i}^\star)
\le
\frac{\lambda}{2}r^{K_{\mathrm{in}}}
\left\|\theta_i^{(0)}-\theta_{\tau_i}^\star\right\|_2^2.
\end{equation}
Because \(\theta_i^{(0)}=\theta^\star+g_{\phi}(z_{\tau_i},\mathbf{Y}_{\tau_i})\),
\begin{equation}
\theta_i^{(0)}-\theta_{\tau_i}^\star
=
\left(m_{\tau_i}-\Delta_{\tau_i}^\star\right)
+\xi_{\tau_i}+\rho_{\tau_i}.
\end{equation}
Taking the squared norm and expectation, the zero-mean and orthogonality conditions in Assumption~\ref{assump:generator_bound} yield
\begin{equation}
\mathbb{E}
\left\|\theta_i^{(0)}-\theta_{\tau_i}^\star\right\|_2^2
\le
\psi_b+\psi_v+\psi.
\end{equation}
Taking expectations in \eqref{eq:excess_risk_init_error} proves \eqref{eq:meta_diffusion_bound}.
\end{proof}

Theorem~\ref{thm:diffusion_meta} gives a conditional local support-loss bound: for a fixed inner-loop budget, adaptation depends on the contraction factor \(r^{K_{\mathrm{in}}}\) and the initialization error \(\psi_b+\psi_v+\psi\). Hence, reducing the environment-conditioned initialization error leads to faster and lower-loss finite-step adaptation.

\begin{corollary}[Upper-Bound Gain from Environmental Conditioning]
\label{cor:env_gain}
Define $R(\mathbf{Y})
:=
\mathbb{E}_{\tau_i}
\left[
\left\|
\mathbb{E}[\Delta_{\tau_i}^\star|\mathbf{Y}_{\tau_i}]
-
\mathbb{E}_{\tau_i}[\Delta_{\tau_i}^\star]
\right\|_2^2
\right]$
as the offset variance explained by the environmental descriptor. Consider an environment-agnostic shared-initialization baseline whose deterministic offset is independent of \(\mathbf{Y}_{\tau_i}\). If the deterministic component of EnvCoLoc is the conditional-mean predictor,
\begin{equation}
\label{eq:conditional_mean_predictor}
m_{\tau_i}=\mathbb{E}[\Delta_{\tau_i}^\star|\mathbf{Y}_{\tau_i}],
\end{equation}
then the difference between the baseline upper bound and the EnvCoLoc upper bound is
\begin{equation}
\label{eq:gain_bound}
\overline{\mathcal{E}}_{\mathrm{SI}}
-
\overline{\mathcal{E}}_{\mathrm{EnvCoLoc}}
=
\frac{\lambda}{2}r^{K_{\mathrm{in}}}
\left[R(\mathbf{Y})-\psi_v-\psi\right].
\end{equation}
Thus EnvCoLoc admits a tighter expected support-loss upper bound than the shared-initialization baseline whenever $R(\mathbf{Y})>\psi_v+\psi$.
\end{corollary}

\begin{proof}
For an environment-agnostic initializer independent of \(\mathbf{Y}_{\tau_i}\), the best constant offset predictor is $c^\star=\mathbb{E}_{\tau_i}[\Delta_{\tau_i}^\star]$.
Its expected initialization error is $\mathbb{E}_{\tau_i}
\left[
\left\|
\Delta_{\tau_i}^\star-
\mathbb{E}_{\tau_i}[\Delta_{\tau_i}^\star]
\right\|_2^2
\right]$.
By the law of total variance,
\begin{align}
&\mathbb{E}_{\tau_i}
\left[
\left\|
\Delta_{\tau_i}^\star-
\mathbb{E}_{\tau_i}[\Delta_{\tau_i}^\star]
\right\|_2^2
\right]
\nonumber\\
&=
\mathbb{E}_{\tau_i}
\left[
\left\|
\Delta_{\tau_i}^\star-
\mathbb{E}[\Delta_{\tau_i}^\star|\mathbf{Y}_{\tau_i}]
\right\|_2^2
\right]
+
R(\mathbf{Y}).
\label{eq:total_variance_decomp}
\end{align}
Under \eqref{eq:conditional_mean_predictor}, $\psi_b
=
\mathbb{E}_{\tau_i}
\left[
\left\|
\Delta_{\tau_i}^\star-
\mathbb{E}[\Delta_{\tau_i}^\star|\mathbf{Y}_{\tau_i}]
\right\|_2^2
\right]$.
Therefore the environment-agnostic baseline has upper bound $\overline{\mathcal{E}}_{\mathrm{SI}}
=
\frac{\lambda}{2}r^{K_{\mathrm{in}}}
(\psi_b+R(\mathbf{Y}))$,
whereas Theorem~\ref{thm:diffusion_meta} gives $\overline{\mathcal{E}}_{\mathrm{EnvCoLoc}}
=
\frac{\lambda}{2}r^{K_{\mathrm{in}}}
(\psi_b+\psi_v+\psi)$.
Subtracting the two bounds gives \eqref{eq:gain_bound}.
\end{proof}

Corollary~\ref{cor:env_gain} compares expected upper bounds rather than guaranteeing that every individual task improves. It states that environmental descriptors are beneficial when the offset variance explained by \(\mathbf{Y}\) is larger than the additional stochastic and generation errors introduced by the generator.

This result further motivates a comparison between stochastic and deterministic environment-conditioned offset generation. We use NoDiff as the deterministic counterpart, which keeps the environmental descriptor but replaces the diffusion generator with a deterministic MLP. Let \(h_{\eta}(\mathbf{Y}_{\tau_i})\) denote this deterministic offset mapper and define its approximation error as $\psi_{\mathrm{det}}
:=
\mathbb{E}
\left\|
h_{\eta}(\mathbf{Y}_{\tau_i})-\Delta_{\tau_i}^\star
\right\|_2^2$.
Applying the same contraction argument as in Theorem~\ref{thm:diffusion_meta}, the NoDiff initialization error is \(h_{\eta}(\mathbf{Y}_{\tau_i})-\Delta_{\tau_i}^\star\), whose expected squared norm is \(\psi_{\mathrm{det}}\). Therefore, $\mathcal{E}_{\mathrm{NoDiff}}
\le
\frac{\lambda}{2}r^{K_{\mathrm{in}}}\psi_{\mathrm{det}}$.
Comparing this bound with Theorem~\ref{thm:diffusion_meta}, EnvCoLoc has a tighter local bound when
\begin{equation}
\label{eq:diff_better_than_det}
\psi_{\mathrm{det}}>\psi_b+\psi_v+\psi.
\end{equation}
Thus, Eq.~\eqref{eq:diff_better_than_det} gives the condition under which diffusion-based generation can improve the bound. This condition does not imply that stochastic generation is universally superior. Rather, it characterizes when the stochastic generator yields a tighter local bound than a deterministic offset mapper.

\begin{remark}[Interpretation of the Diffusion Residual]
\label{rem:diffusion_psi}
Under Assumption~\ref{assump:score}, the residual term \(\psi\) can be interpreted as the part of the initialization error caused by imperfect diffusion generation. It contains two sources: the denoising network may not exactly match the target score, and the reverse process is implemented with a finite number of numerical steps. This can be summarized as
\begin{equation}
\label{eq:psi_decomp}
\psi
\le
C_1\epsilon_{\mathrm{score}}^2
+
C_2\Delta t^2,
\end{equation}
where \(C_1,C_2>0\) are problem-dependent constants and \(\Delta t\) is the reverse-time discretization step. A more accurate denoising model or a finer reverse-time discretization reduces this residual. This relation is used as an interpretive decomposition of the generation error, consistent with standard score-based generative modeling analyses.

Under the conditional-mean case in Corollary~\ref{cor:env_gain}, the deterministic bias term becomes
\begin{equation}
\label{eq:conditional_variance}
\psi_b
=
\sigma_{\Delta|\mathbf{Y}}^2
:=
\mathbb{E}_{\tau_i}
\left[
\left\|
\Delta_{\tau_i}^\star-
\mathbb{E}[\Delta_{\tau_i}^\star|\mathbf{Y}_{\tau_i}]
\right\|_2^2
\right].
\end{equation}
Substituting \eqref{eq:psi_decomp} into Theorem~\ref{thm:diffusion_meta} yields
\begin{equation}
\label{eq:full_bound}
\begin{aligned}
\mathcal{E}_{\mathrm{EnvCoLoc}}
&\le
\frac{\lambda}{2}r^{K_{\mathrm{in}}}
\left(
\sigma_{\Delta|\mathbf{Y}}^2
+
\psi_v
+
C_1\epsilon_{\mathrm{score}}^2
+
C_2\Delta t^2
\right).
\end{aligned}
\end{equation}
Equation~\eqref{eq:full_bound} separates three factors: environment-conditioned residual variation, stochastic sampling variation, and diffusion-generation error. The inner-loop adaptation contracts these terms by \(r^{K_{\mathrm{in}}}\).
\end{remark}

\section{Experimental Results}
\label{sec:ER}

\subsection{Experimental Setup}
\label{subsec:es}
This section presents the data collection platform, experimental scenarios, and baseline models used for comparison. The datasets are collected from two real-world indoor environments through site surveys conducted on the campus of The Chinese University of Hong Kong, Shenzhen.

\begin{figure*}
	\centering
	\subfigure[]{
		\includegraphics[scale=0.08]{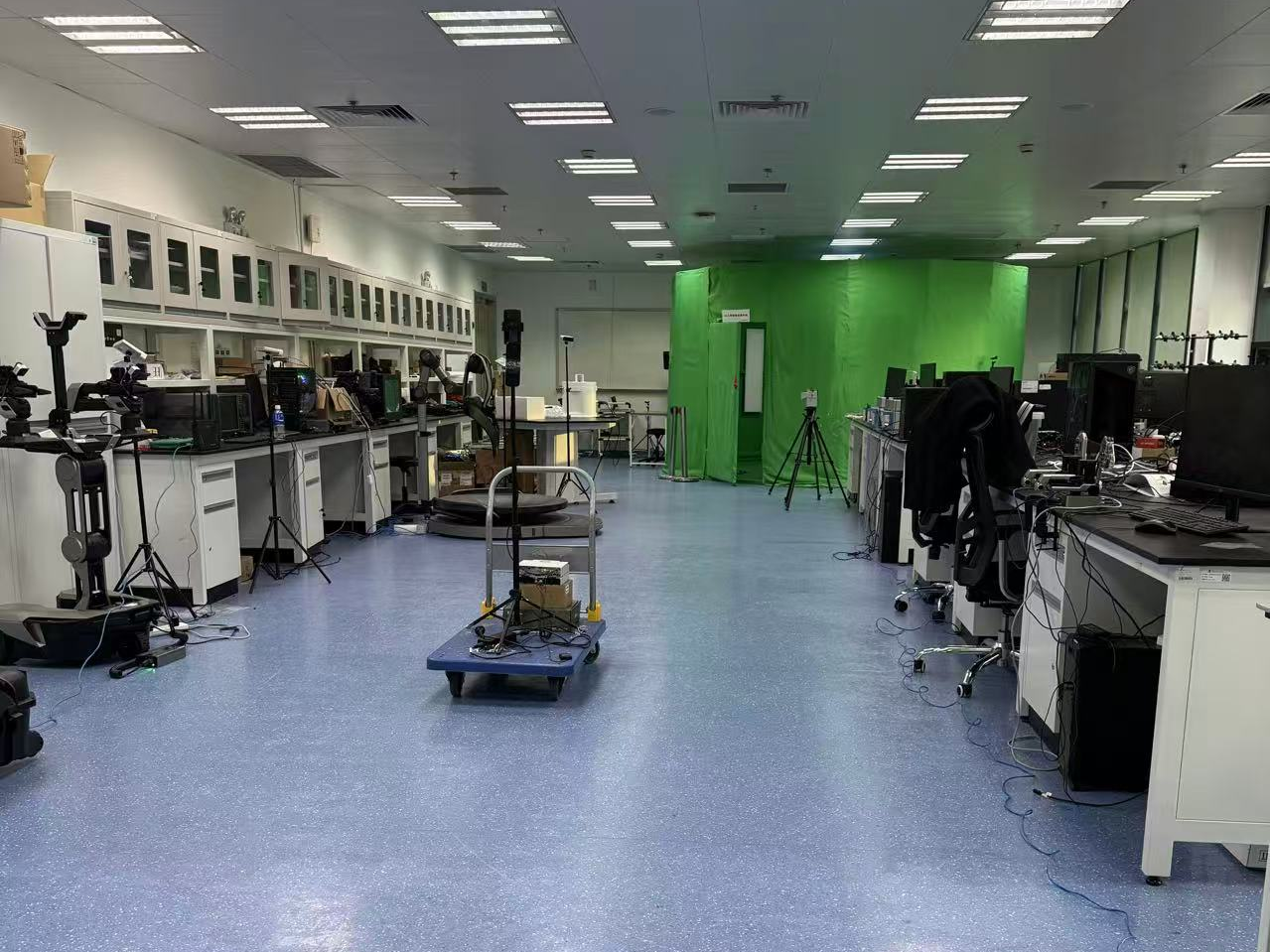}}
		\subfigure[]{
		\includegraphics[scale=0.16]{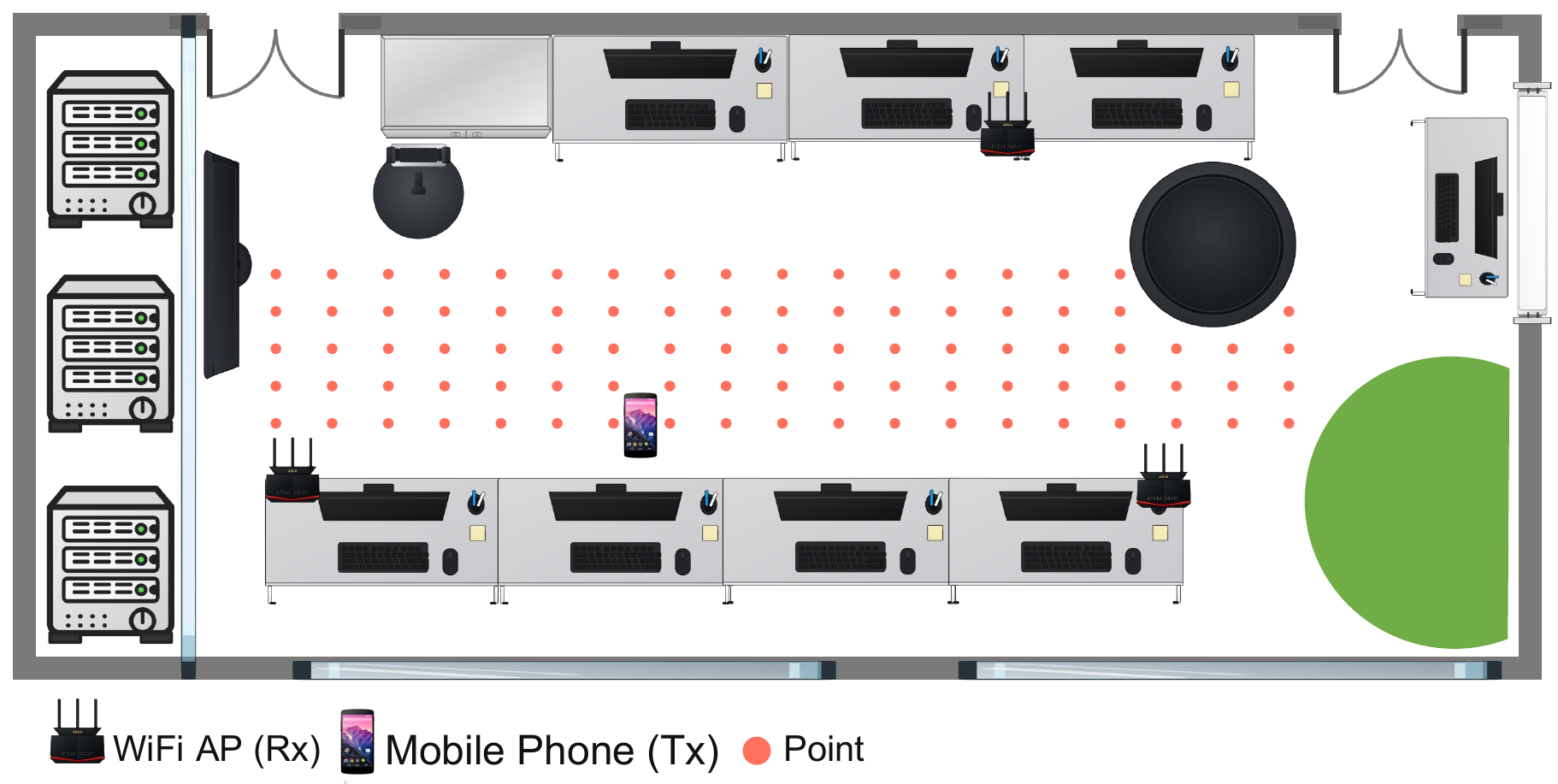}}
		\subfigure[]{
		\includegraphics[scale=0.06]{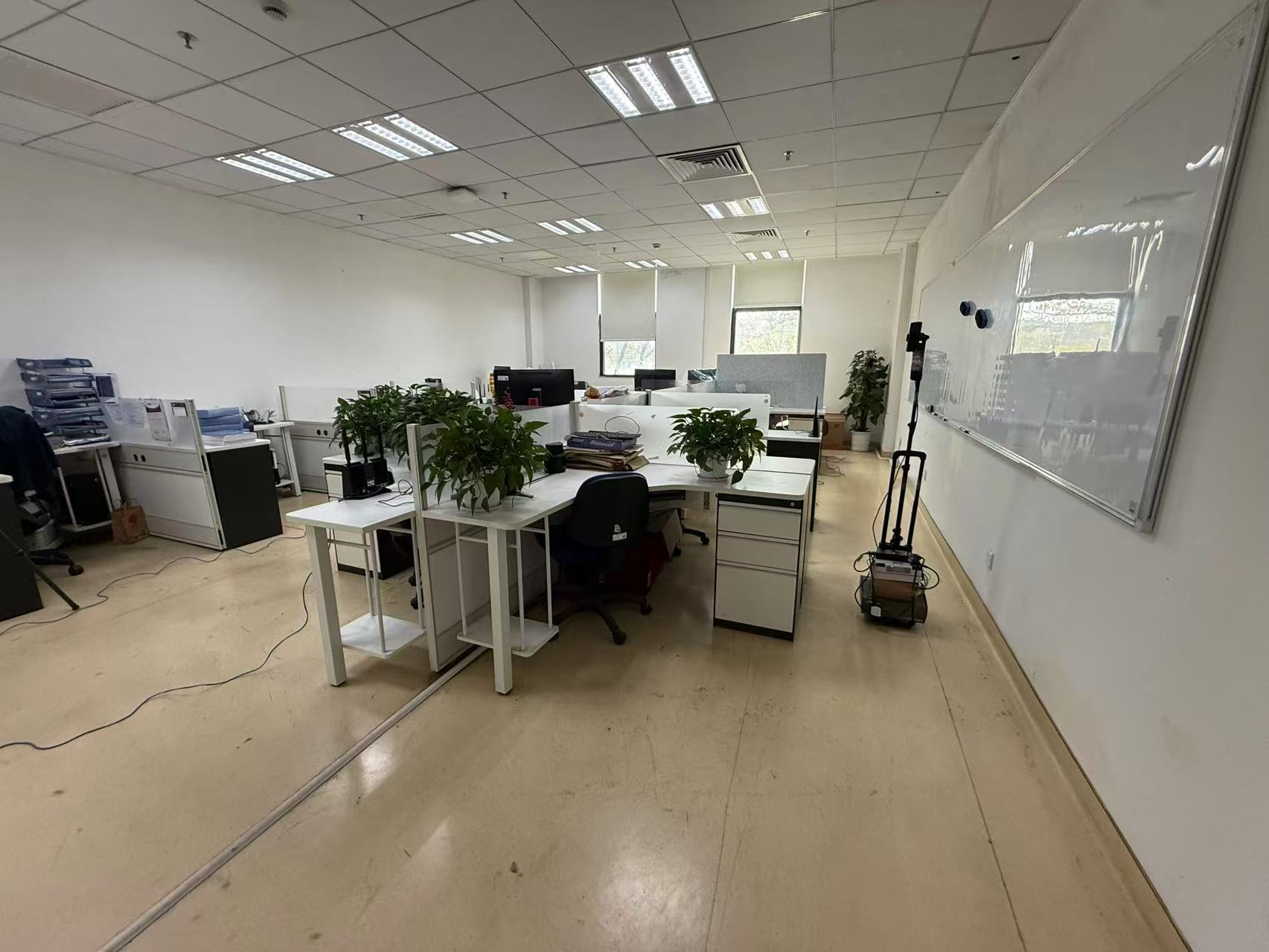}}
		\subfigure[]{
		\includegraphics[scale=0.165]{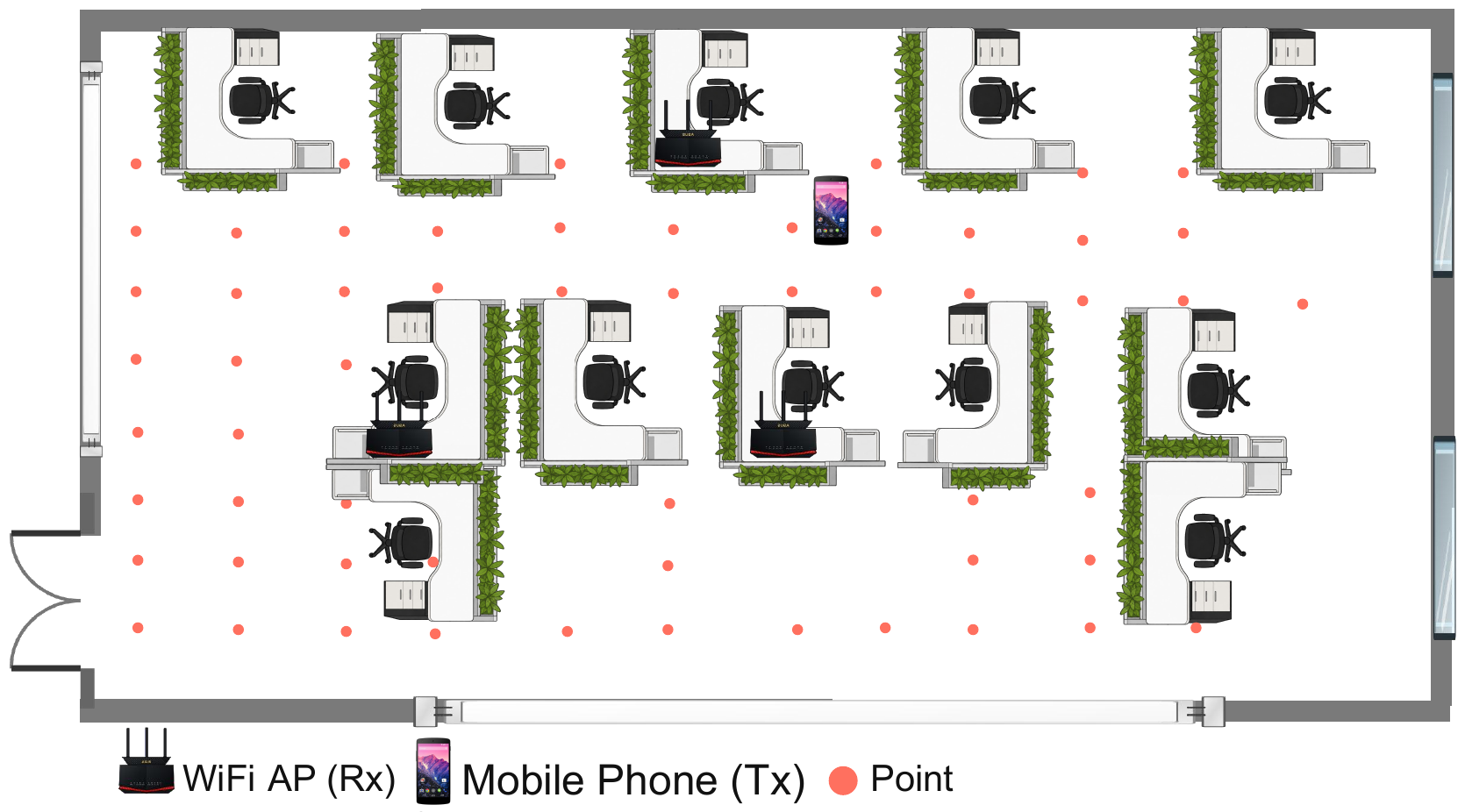}
	}
		\caption{Photographs and layouts of the two experimental scenarios. The lab is denoted as S1 and is dominated by LOS links, while the office is denoted as S2 and is dominated by NLOS links. The photographs in (a) and (c) correspond to deploy1, where no additional obstacles are introduced. (a) Photograph of S1. (b) Layout of S1. (c) Photograph of S2. (d) Layout of S2.}
	\label{fig:scenarios}
\end{figure*}

\subsubsection{Data Collection Platform}
We built an uplink communication platform consisting of $Z=3$ WiFi APs connected to a central server via a network switch. Each AP is assigned a unique IP address, enabling independent and simultaneous data reception from the mobile phone. Data collection is performed using the Nexmon CSI Extractor Tool~\cite{nexmon:project}, an open-source framework for per-frame CSI extraction. It supports up to four spatial streams and four receiving chains on Broadcom and Cypress WiFi chips, with bandwidths up to 80 MHz in both the 2.4 GHz and 5 GHz bands, and is compatible with devices ranging from Raspberry Pi platforms to smartphones and commercial WiFi APs~\cite{gringoli2019free}. 

In the implementation, a Nexus 5 smartphone is used as the transmitter, while ASUS RT-AC86U routers act as the receivers. The Nexmon framework~\cite{nexmon:project} is used to patch a modified firmware~\cite{gringoli2019free}, so that each collected packet contains both \ac{rssi} and CSI measurements. During CSI collection, the WiFi radio remains active for packet transmission, while the control/backhaul connection between the APs and the server is maintained through Ethernet. The system operates under the 802.11ac standard in the 5 GHz band on channel 157 at 5785 MHz with an 80 MHz bandwidth.

\subsubsection{Experimental Scenarios}
As illustrated in Fig.~\ref{fig:scenarios}, the experiments are conducted in two representative indoor propagation scenarios: a lab scenario (S1) dominated by \ac{los} links and an office scenario (S2) dominated by \ac{nlos} links. The lab covers an area of 12.4~m $\times$ 8.6~m, while the office covers an area of 8.6~m $\times$ 6.2~m and contains desks, walls, and other obstacles that create richer multipath propagation. In both scenarios, a Nexus 5 smartphone is used as the mobile transmitter at a height of approximately 1.4~m, and three WiFi APs are used as receivers at a height of approximately 1~m. Survey points are sampled with a spacing of 0.4~m in the lab and 0.6~m in the office, resulting in 90 survey locations in the lab and 60 survey locations in the office. The number of valid CSI packets varies slightly across AP links and deployment rounds, with about 900--1000 packets recorded for each survey point from each AP. Across all deployment rounds, the processed calibrated CSI dataset contains 766,690 AP-link packets in the lab and 538,929 AP-link packets in the office. In each scenario, data collection is performed across three deployment
rounds with different obstacle configurations. Deploy1 corresponds to
the original scenario photographs in Fig.~\ref{fig:scenarios}, where no
additional obstacles are introduced. In deploy2, cardboard boxes are
added to change the obstacle distribution. In deploy3, additional
obstacles are introduced and metal surfaces are attached to selected
cardboard boxes, which changes both the spatial distribution and the
material composition of propagation-relevant obstacles. Fig.~\ref{fig:obstacle-variation}
shows the added obstacle configurations in deploy2 and deploy3. For
each deployment round, the point-cloud map is updated and the
corresponding material-aware descriptors are recomputed.

\begin{figure}[t]
	\centering
	\subfigure[Lab deploy2]{\includegraphics[width=0.48\columnwidth]{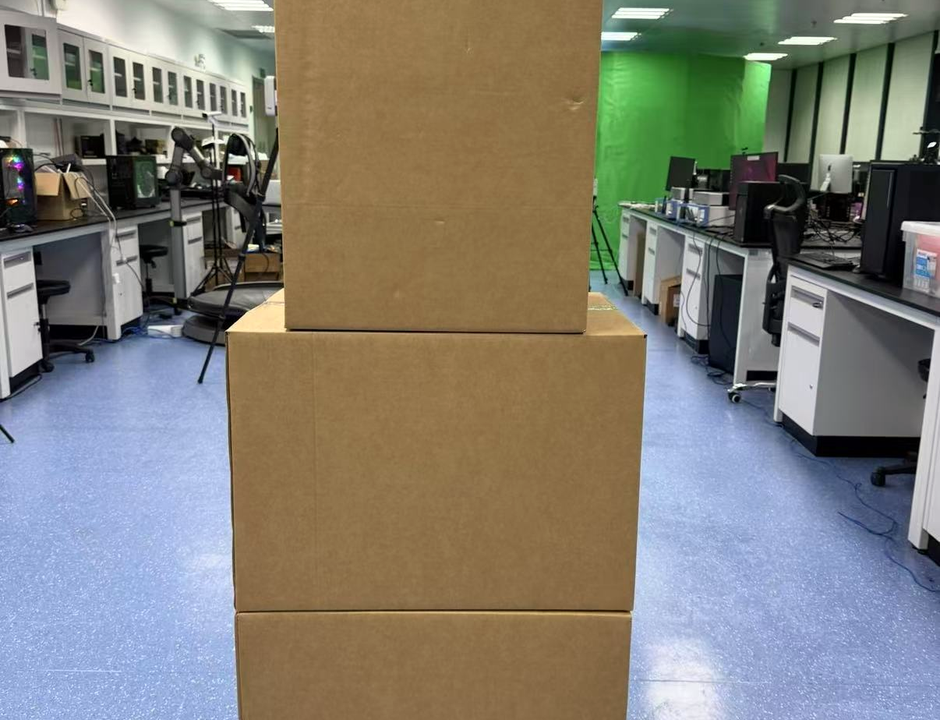}}
	\subfigure[Lab deploy3]{\includegraphics[width=0.48\columnwidth]{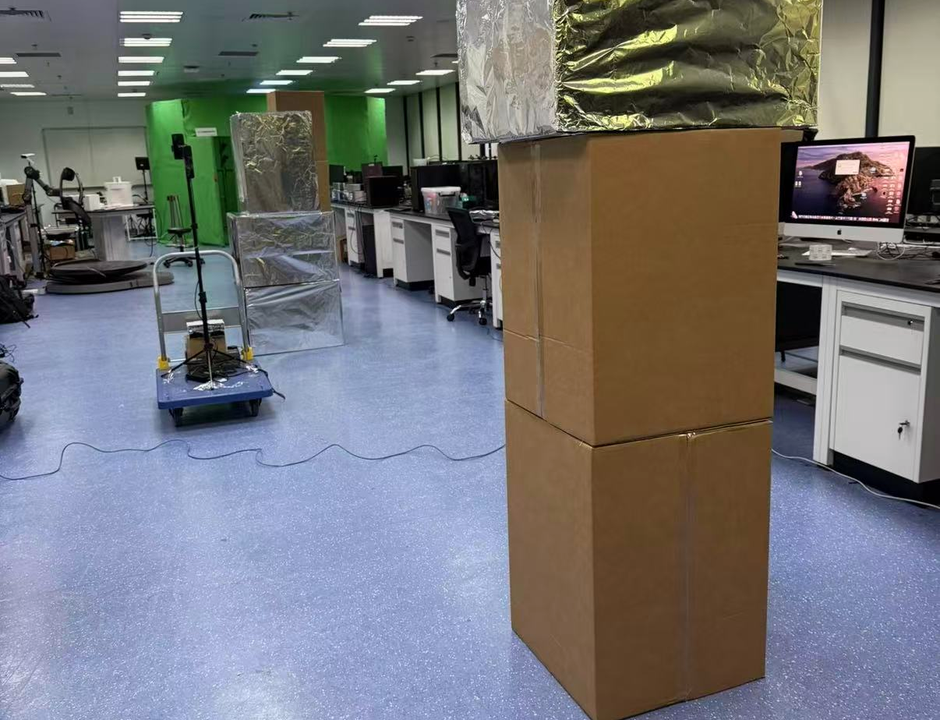}}
	\subfigure[Office deploy2]{\includegraphics[width=0.48\columnwidth]{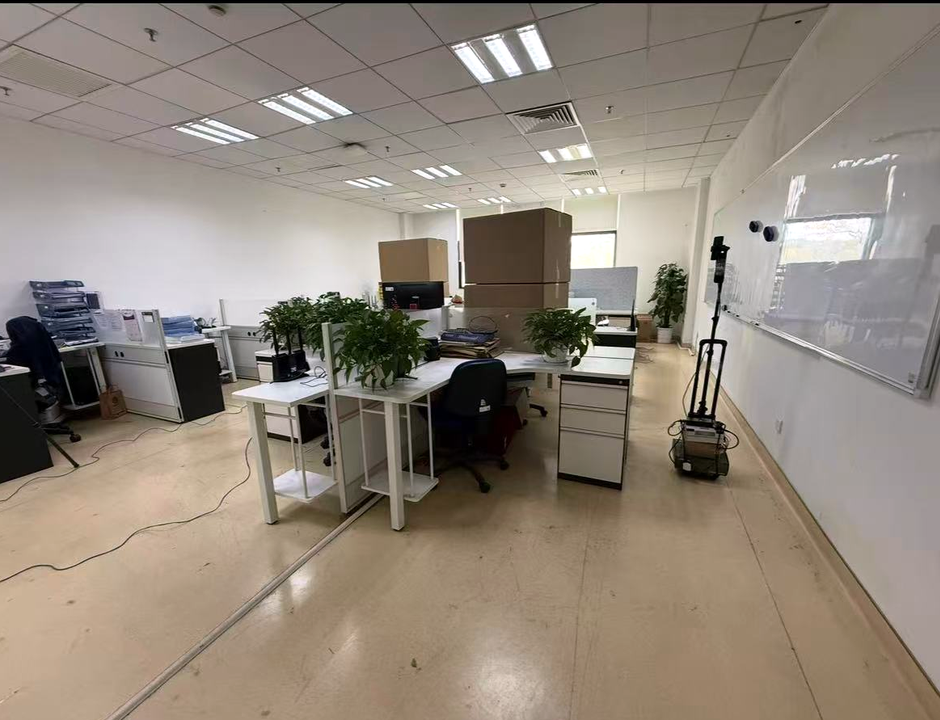}}
	\subfigure[Office deploy3]{\includegraphics[width=0.48\columnwidth]{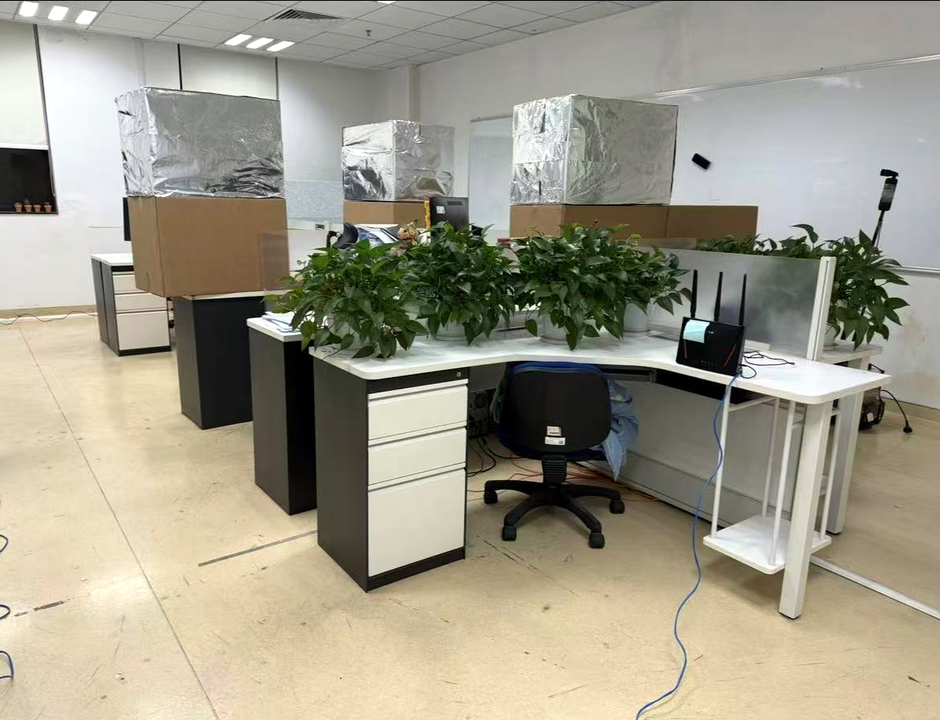}}
	\caption{Added obstacle configurations in deploy2 and deploy3. Deploy1 corresponds to the original scenario photographs in Fig.~\ref{fig:scenarios}, where no additional obstacles are introduced.}
	\label{fig:obstacle-variation}
\end{figure}

Each meta-learning task $\tau_i$ corresponds to a randomly sampled
sub-region within a given scenario and is treated as a local environment.
In the current evaluation, data are organized by deployment rounds.
For the lab scenario, deploy1 and deploy2 are used as historical
environments for meta-training, while deploy3 is used as the target
environment for fine-tuning and testing. For the office scenario,
deploy1 and deploy3 are used as historical environments, while deploy2
is used as the target environment. Unless otherwise stated, the
support-set size is $N_s = 10$. The support-size study further evaluates
$N_s \in \{1,5,10,30\}$. Query samples are always disjoint from the
support samples; the lab evaluation uses $N_q=20$, and the office
target evaluation uses 40 query samples from deploy2. Meta-training
uses 1200 outer-loop iterations with randomly sampled task-level
sub-regions. The point-cloud maps used to compute the environmental
descriptors were captured using a Leica RTC360 LT terrestrial laser
scanner. Fig.~\ref{fig:point-cloud} presents one representative point-cloud image for each
scenario. We further assign coarse material-category labels using
CloudCompare~\cite{cloudcompare}. The current material-aware
descriptor uses six coarse labels: cement, metal, wood, plastic, tile,
and glass.
\begin{figure}[t]
	\centering
	\includegraphics[width=\columnwidth]{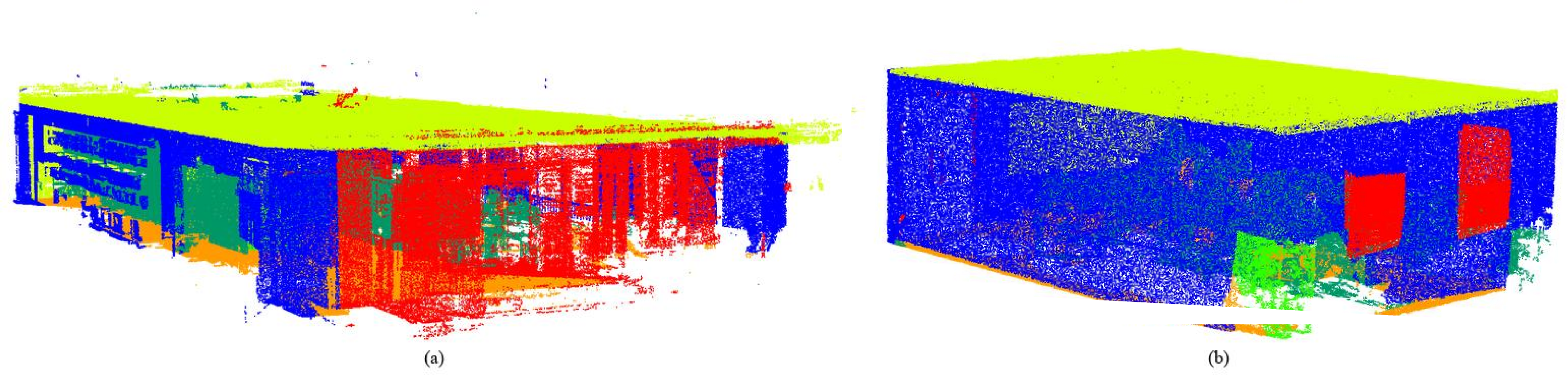}
	\caption{(a) Point cloud representation in the lab. (b) Point cloud representation in the office. Obstacles are manually assigned coarse material-category labels and shown in different colors.}~\label{fig:point-cloud}
\end{figure}

\subsubsection{Baselines and Settings}
The environment-conditioned diffusion generator uses a lightweight U-Net with two downsampling layers of 64 and 128 channels, a cross-attention bottleneck, two upsampling layers of 128 and 64 channels, and a final $1\times 1$ convolution layer. The cross-attention projection dimension is 64, the latent channel width is 64, and the generator contains approximately 0.5M trainable parameters. The reverse diffusion process uses $T_{\mathrm{d}}=20$ denoising steps with a linear noise schedule, where the diffusion noise rates are $\beta_1^{\mathrm{d}}=10^{-4}$ and $\beta_{T_{\mathrm{d}}}^{\mathrm{d}}=0.02$. All experiments are conducted on an NVIDIA A100 GPU with 80 GB memory.

The localization network is a four-layer fully connected model with hidden dimensions $[256,128,64]$ and ReLU activations, mapping concatenated calibrated CSI amplitudes from all APs to 3D coordinates. During meta-training, we set the inner-loop learning rate to $\alpha=0.008$, the outer-loop learning rate to $\beta=0.001$, the number of inner-loop steps to $K_{\mathrm{in}}=5$, and the diffusion denoising steps to $T_{\mathrm{d}}=20$. Consistent with Eq.~\eqref{eq:meta_obj}, the shared localization initialization $\theta$ and the diffusion generator $\phi$ are jointly optimized by the outer-loop meta-loss. 

We compare the proposed EnvCoLoc framework against the following methods:
\begin{itemize}
	\item {Random Initialization (RI)}: We randomly generate a set of network parameters for initializing the new environment. The task format and the hyperparameters of the neural network remain the same as those used in EnvCoLoc.
	\item \ac{knn}~\cite{bahl2000radar}: We adopt the Euclidean distance metric to select the $K_{\rm NN}$ RPs closest to the target fingerprint in the signal space, where $K_{\rm NN}$ is set equal to the support-set size used at the
meta-test stage. The averaged locations of the selected RPs are then treated as the estimation result.
	\item {ConFi~\cite{chen2017confi}}: ConFi is the first method using CNNs to learn CSI fingerprints for indoor localization. We use the same dataset split as the proposed method, where data from historical environments serve as the training set and data collected in the new environment serve as the test set.
	\item {TransLoc~\cite{li2021transloc}}: TransLoc employs transfer learning to map cross-domain CSI features into a shared feature space for localization. The model is pretrained on historical environments and fine-tuned on the new environment with the same support samples.
\end{itemize}

The compared baselines follow their original formulations and do not use environmental descriptors. For fair comparison, all learning-based methods use the same localization-network input, namely the calibrated CSI amplitudes from all APs, $\mathbf{x} \in \mathbb{R}^{Z \times 256}$. In EnvCoLoc, the environmental descriptors are used only to condition the diffusion generator for meta-initialization. The respective effects of environmental conditioning and diffusion-based offset generation are further examined through the NoEnv and NoDiff ablations in Sec.~V-D.

\begin{figure}[t]
	\centering
	\includegraphics[width=0.9\columnwidth]{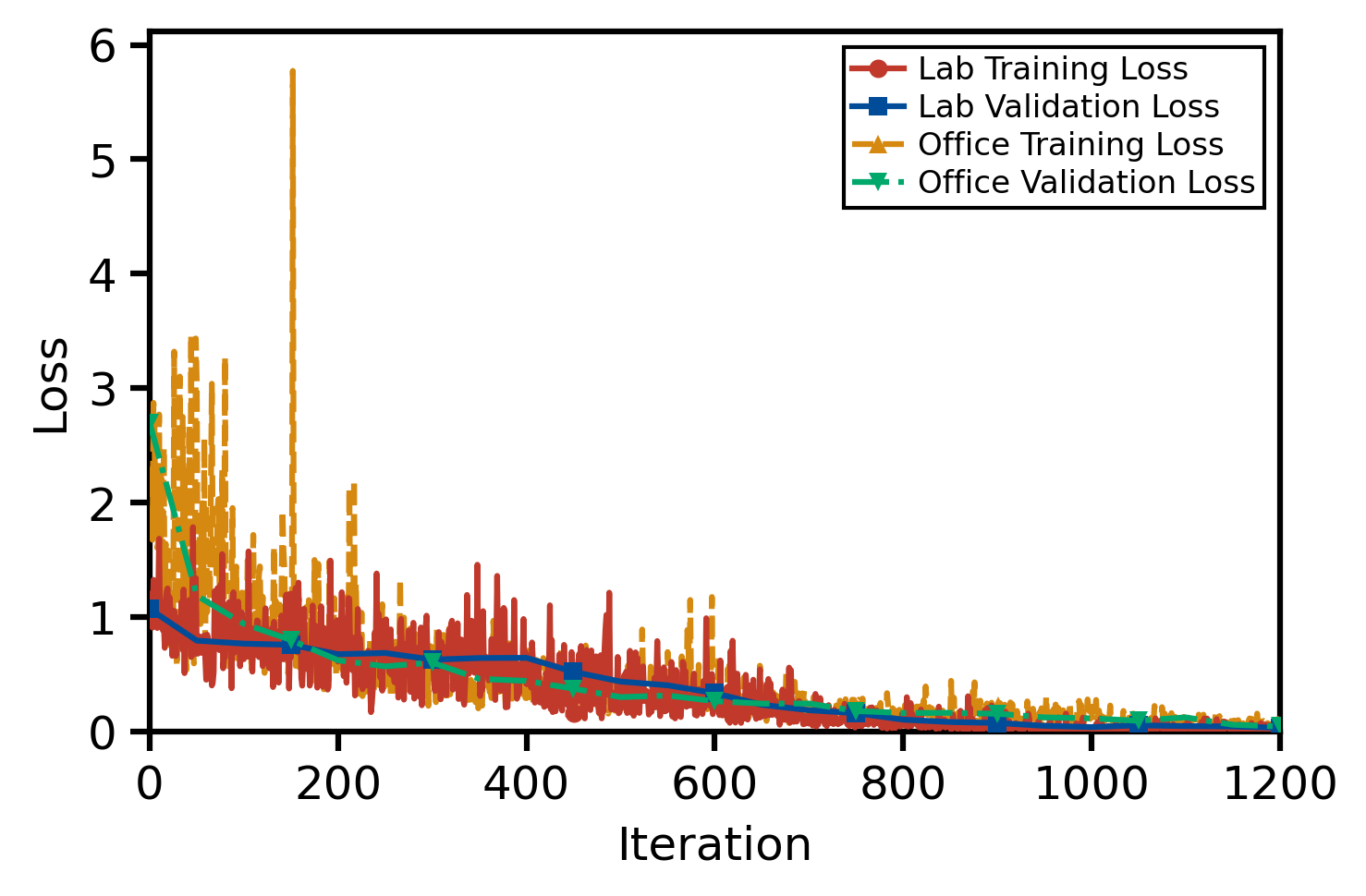}
	\caption{The convergence comparison between training loss and validation loss during the meta-training stage. Specifically, the red and orange curves represent the training loss for the lab and office scenarios, respectively, while the blue and green curves show the validation loss for the lab and office scenarios. All curves converge, indicating effective model training without overfitting.}	\label{fig:meta-train-loss}
\end{figure}

\subsection{Convergence Analysis}
To assess the role of meta-parameter initialization, we first examine the convergence behavior in new environments. Fig.~\ref{fig:meta-train-loss} compares the training and validation loss convergence during the meta-training stage for the lab (S1) and office (S2) scenarios. In both cases, the losses gradually converge, indicating stable joint optimization of the localization initialization and the diffusion generator. Compared with the lab scenario, the office scenario exhibits larger fluctuations during training and converges to a higher loss level. This is consistent with its \ac{nlos}-dominant propagation, where denser obstacles induce stronger multipath variations and increase the environment-conditioned residual variation.

When the environments change, Fig.~\ref{fig:meta-test-loss} (a) and (b) show the meta-test convergence in the lab and office scenarios, respectively, with only 10 support samples in the new environment. The solid red and blue curves denote the fine-tuning loss and test loss obtained from the proposed meta-parameter initialization, whereas the dashed curves correspond to random initialization. In both scenarios, meta-parameter initialization rapidly reduces the fine-tuning loss within the first few update steps and maintains a substantially lower test loss throughout adaptation. In Fig.~\ref{fig:meta-test-loss} (a), random initialization yields a slowly decreasing training loss but its test loss remains high and nearly flat, indicating poor adaptation under sparse support data. In Fig.~\ref{fig:meta-test-loss} (b), random initialization also decreases only gradually and stays well above the meta-initialized curves for both training and test loss. These observations agree with Lemma~\ref{lem:adapt} and Theorem~\ref{thm:diffusion_meta}: when the initialization $\theta_i^{(0)}$ is closer to the task-specific optimum $\theta_{\tau_i}^\star$, the remaining adaptation error after $K_{\mathrm{in}}$ updates is contracted by $r^{K_{\mathrm{in}}}$.

\begin{figure}
	\centering
	\subfigure[]{\includegraphics[width=0.241\textwidth]{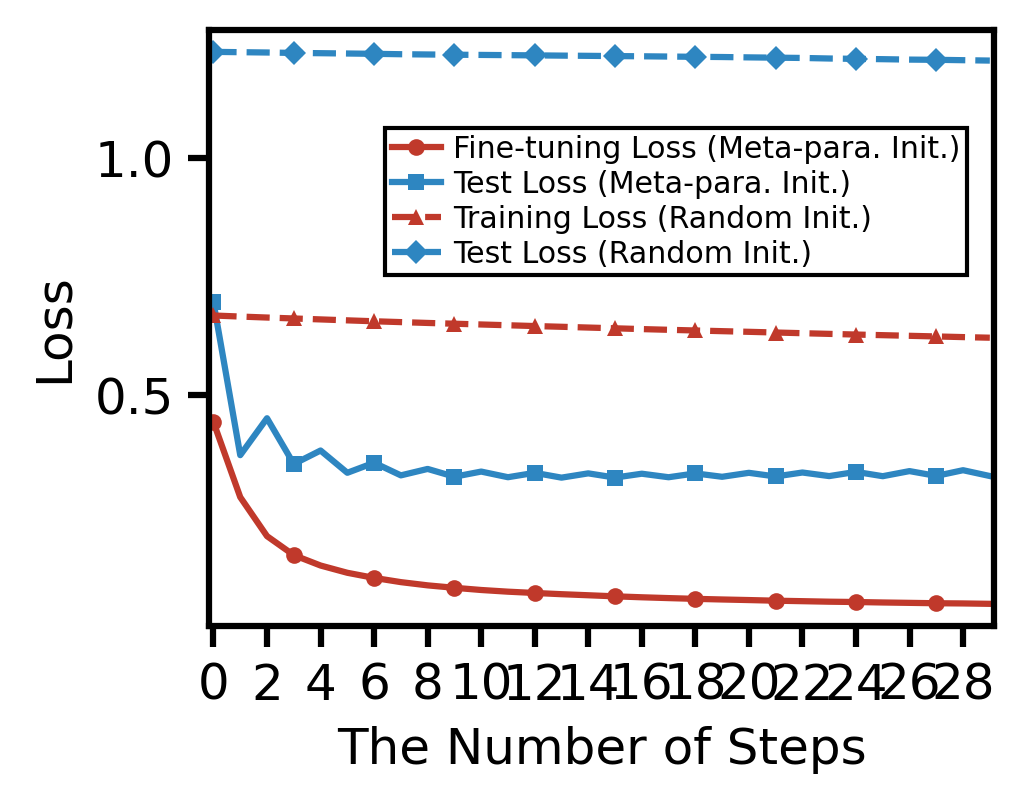}}
	\subfigure[]{\includegraphics[width=0.241\textwidth]{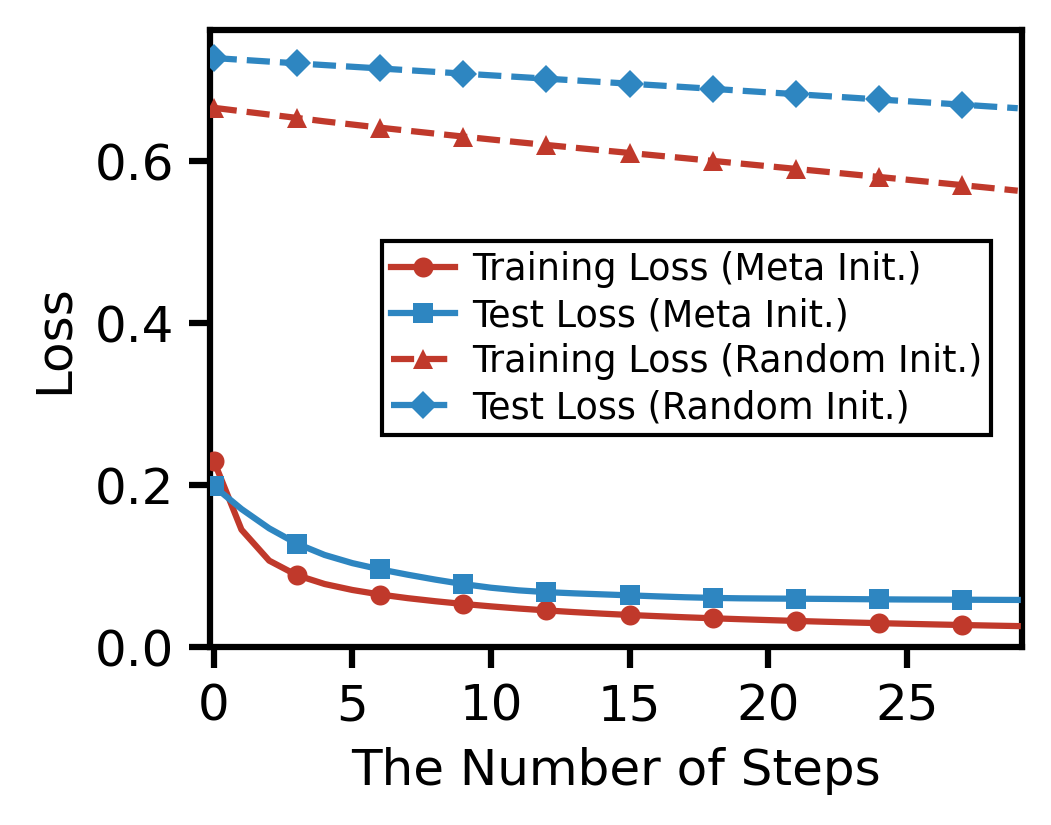}}
	\caption{(a) Loss comparison during the meta-test stage in the lab scenario. (b) Loss comparison during the meta-test stage in the office scenario. The solid red and blue curves represent the fine-tuning and test losses with meta-parameter initialization, while the dashed red and blue curves represent the corresponding losses with random initialization. Meta-parameter initialization leads to rapid loss reduction and consistently lower test loss, whereas random initialization adapts slowly and remains at a much higher loss level in both scenarios.}
	\label{fig:meta-test-loss}
\end{figure}

\subsection{Localization Performance}
We next evaluate the effect of environmental information on localization accuracy. Since the mobile-phone height is fixed during data collection, we report the horizontal Euclidean localization error $e=\sqrt{(\hat{x}-x)^2+(\hat{y}-y)^2}$. Fig.~\ref{fig:cdf}(a) and (b) compare different methods in the lab and office scenarios, respectively, when only 10 support samples are available in the new environment. EnvCoLoc achieves the best overall performance, with a mean localization error of 0.82~m in the lab scenario (\ac{los}) and 1.08~m in the office scenario (\ac{nlos}). Compared with MetaLoc, which does not use the environment-conditioned diffusion offset generator, EnvCoLoc reduces the mean error by 26.1\% in the lab scenario and 24.9\% in the office scenario. The improvement is also visible in the tail distribution, where EnvCoLoc reduces the 80th-percentile error from 1.56~m to 1.09~m in the lab scenario and from 1.79~m to 1.51~m in the office scenario. This behavior is qualitatively consistent with Corollary~\ref{cor:env_gain}: when the point-cloud descriptor explains a non-negligible portion of the task-dependent offset variance, i.e., when $R(\mathbf{Y})$ is sufficiently large, conditioning the initialization on $\mathbf{Y}$ tightens the excess-loss upper bound relative to an environment-agnostic initializer. TransLoc and ConFi yield larger errors than MetaLoc, while \ac{knn} and random initialization are less stable under sparse support samples, highlighting the advantage of environment-conditioned meta-initialization for data-efficient adaptation.

Tab.~\ref{tab:results} summarizes the mean, median, and 80th-percentile localization errors of all methods with 10 support samples. EnvCoLoc achieves the lowest error across all metrics and both scenarios, which is consistent with the CDF curves in Fig.~\ref{fig:cdf}.

\begin{table}[t]
	\centering
	\caption{Localization Error (m) with 10 Support Samples}
	\label{tab:results}
	\begin{tabular}{l|ccc|ccc}
		\toprule
		& \multicolumn{3}{c|}{lab (\ac{los})} & \multicolumn{3}{c}{office (\ac{nlos})} \\
		Method & Mean & Med. & 80th & Mean & Med. & 80th \\
			\midrule
			RI         & 2.11 & 2.18 & 3.17 & 2.64 & 2.70 & 3.32 \\
			\ac{knn}        & 2.41 & 2.36 & 3.54 & 3.64 & 3.83 & 4.23 \\
			TransLoc   & 2.31 & 2.12 & 3.74 & 2.54 & 2.65 & 3.10 \\
			ConFi      & 2.18 & 2.02 & 2.97 & 2.06 & 1.75 & 3.06 \\
			MetaLoc    & 1.10 & 1.00 & 1.56 & 1.44 & 1.42 & 1.79 \\
			EnvCoLoc   & \textbf{0.82} & \textbf{0.79} & \textbf{1.09} & \textbf{1.08} & \textbf{0.95} & \textbf{1.51} \\
		\bottomrule
	\end{tabular}
\end{table}

Fig.~\ref{fig:cdf_across_sizes} studies how the number of labeled support samples affects meta-test performance. In the lab scenario (\ac{los}), increasing the support size from $N_s=1$ to $N_s=5$ reduces the mean error from 1.68~m to 0.90~m, and $N_s=10$ further reduces it to 0.82~m. The result with $N_s=30$ remains competitive, with a mean error of 0.97~m, although it is slightly worse than $N_s=10$ due to support-set sampling variation. In the office scenario (\ac{nlos}), the mean errors for $N_s=1,5,10,30$ are 1.03~m, 1.04~m, 1.08~m, and 1.02~m, respectively, while the 80th-percentile error decreases from 1.61~m at $N_s=1$ to 1.30~m at $N_s=30$. These results indicate that EnvCoLoc already adapts effectively with a small support set, and additional samples mainly improve the upper tail of the error distribution rather than changing the mean error monotonically in every run. From the perspective of Eq.~\eqref{eq:full_bound}, the support set affects the realized finite-step adaptation around the generated initialization, whereas the environmental descriptor mainly reduces the initialization-dependent term before adaptation. This explains why the method remains effective even when $N_s$ is small, while larger support sets primarily refine difficult tail cases.

\begin{figure}
	\centering
	\subfigure[]{\includegraphics[width=0.241\textwidth]{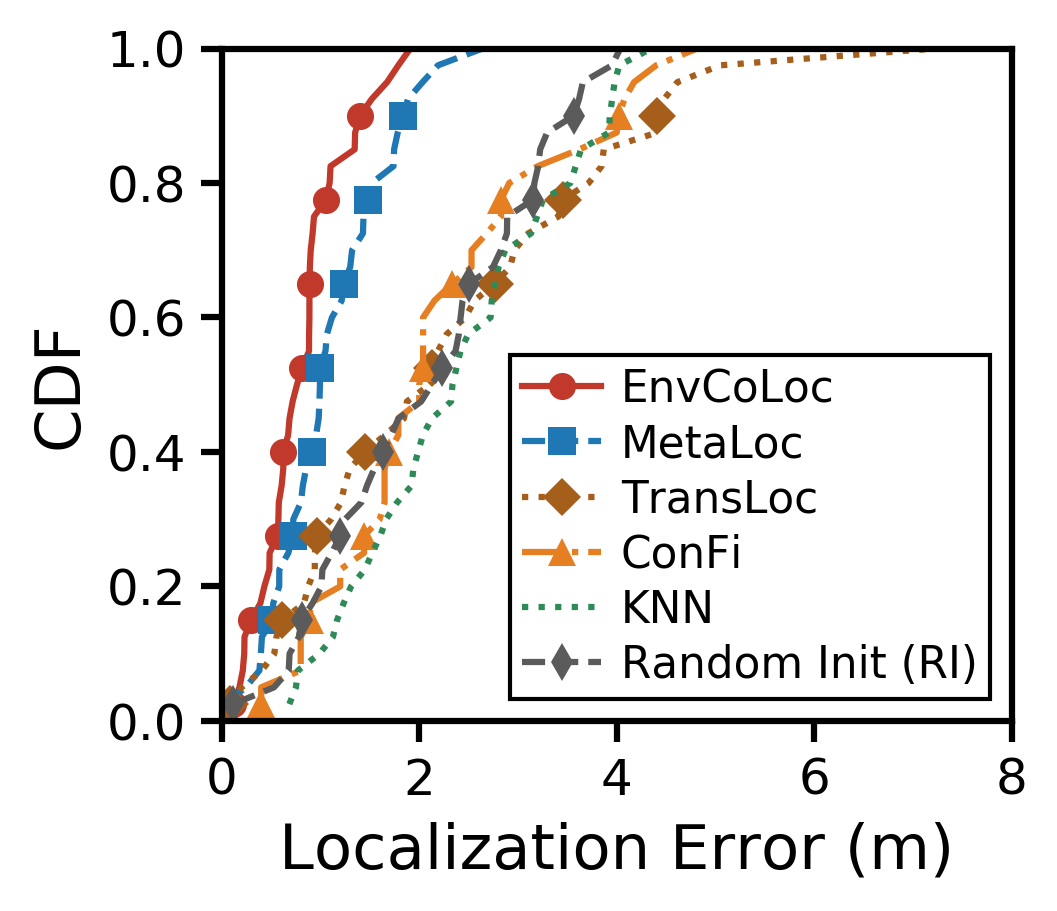}}
	\subfigure[]{\includegraphics[width=0.241\textwidth]{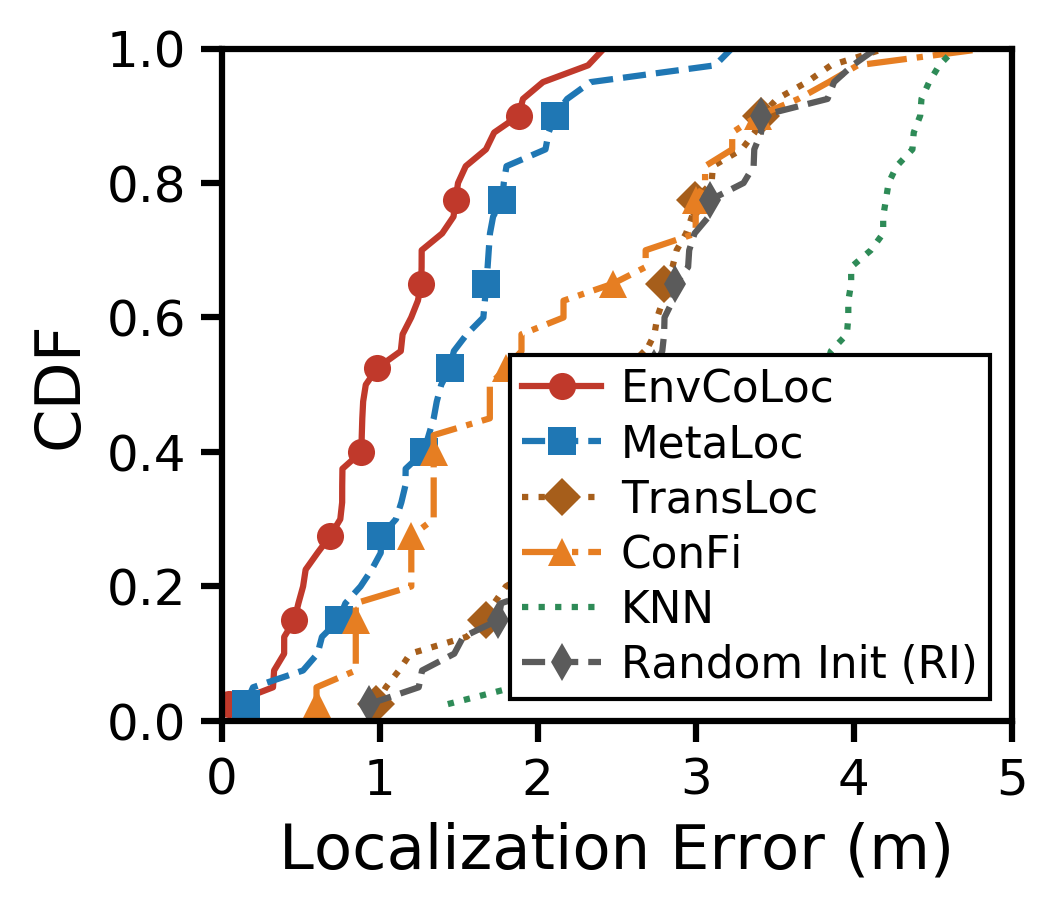}}
	\caption{Comparison of CDFs of localization errors across different methods with 10 support samples in the new environment. (a) lab, \ac{los}. (b) office, \ac{nlos}.}~\label{fig:cdf}
\end{figure}

\begin{figure}
	\centering
	\subfigure[]{
		\includegraphics[width=0.225\textwidth]{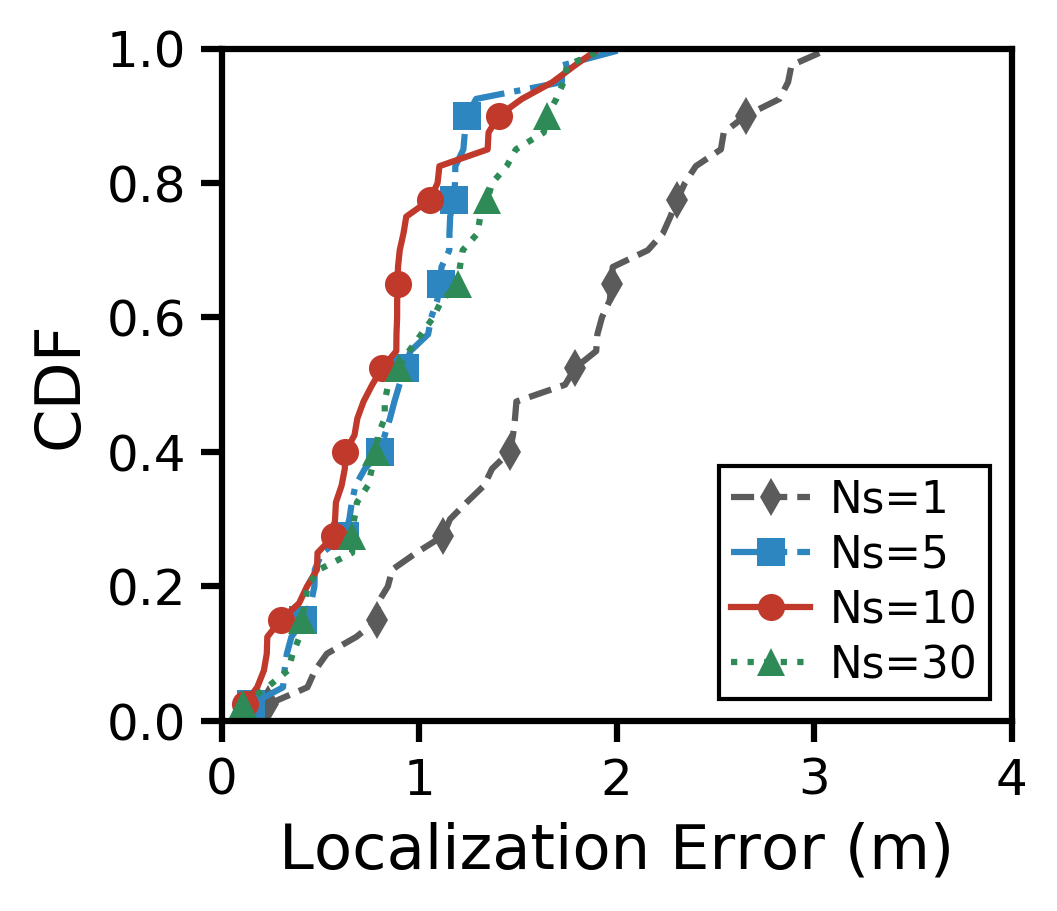}}
		\subfigure[]{
		\includegraphics[width=0.225\textwidth]{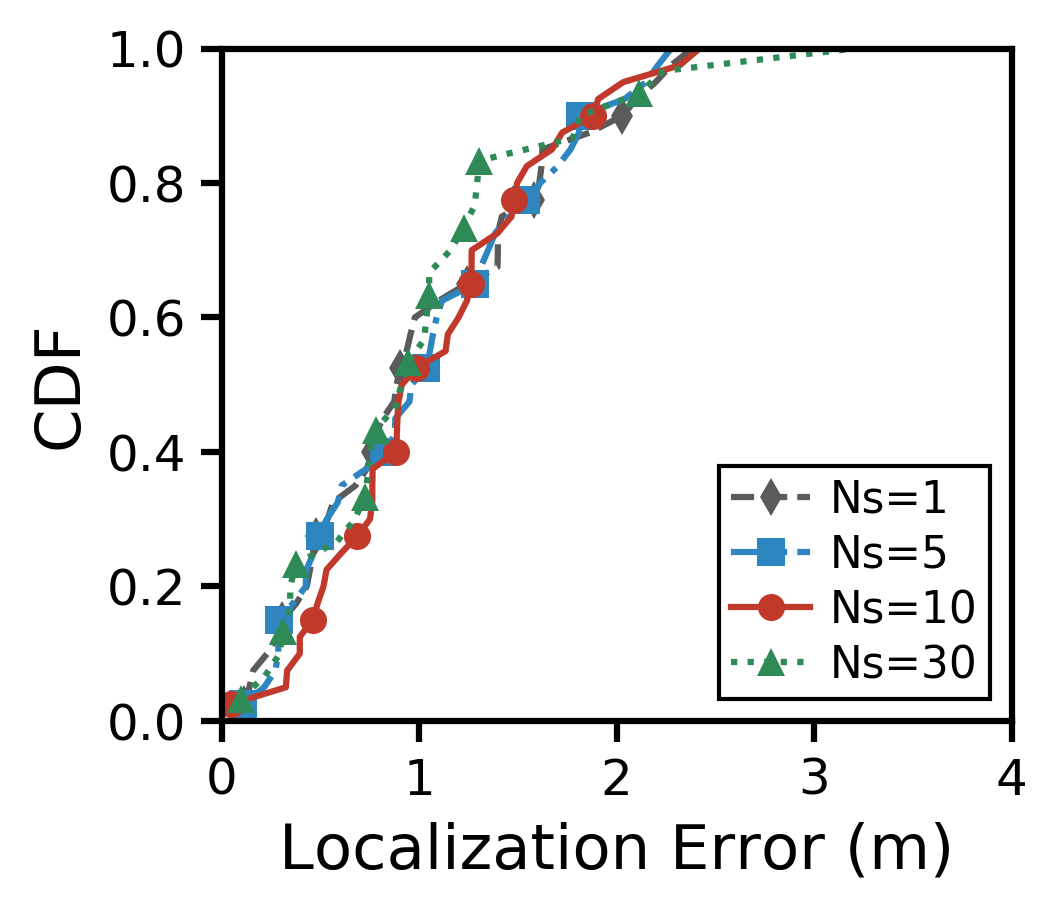}
	}
	\caption{CDF of localization error under different support set sizes at the meta-test stage for the lab and office scenarios.}
	\label{fig:cdf_across_sizes}
\end{figure}

\subsection{Ablation Study}
\label{subsec:ablation}
To isolate the contribution of each component, we evaluate three ablation variants:
\begin{itemize}
	\item \textbf{NDNE} (No Diffusion, No Environment): removes both the diffusion generator and environmental descriptors; the model reduces to a standard meta-learning baseline with the same network architecture.
	\item \textbf{NoDiff} (No Diffusion): replaces the diffusion generator with a deterministic MLP that maps $\mathbf{Y}$ to $\Delta\theta$, retaining the environmental input but removing stochastic offset generation.
	\item \textbf{NoEnv} (No Environment): uses the diffusion generator to produce offsets without environmental conditioning, so the offset becomes $\Delta\theta = g_\phi(z)$ without $\mathbf{Y}$.
\end{itemize}

Unless otherwise stated, all ablation variants use $N_s=10$ and $K_{\mathrm{in}}=5$. For consistency across scenarios, NDNE follows the same MetaLoc-style initialization protocol in both the lab and office experiments, but removes both the environmental descriptor and the diffusion offset. In the current implementation, NoEnv uses the diffusion generator without the material descriptor, while NoDiff uses a deterministic MLP offset generator with a comparable training protocol. The deterministic MLP in NoDiff is configured with a comparable number of trainable parameters to the diffusion generator and is trained with the same outer-loop meta-objective.

Fig.~\ref{fig:ablation_cdf}(a) and (b) present the CDFs of localization errors for all variants with 10 support samples, and Tab.~\ref{tab:ablation} reports the corresponding statistics. NDNE yields the largest mean error among the four variants in both scenarios, with mean errors of 1.10~m in the lab scenario and 1.44~m in the office scenario. This is consistent with the SI baseline in Corollary~\ref{cor:env_gain}: when both environment conditioning and diffusion generation are removed, the initializer loses the variance-reduction benefit associated with $R(\mathbf{Y})$. In the office scenario, NoDiff exhibits the largest 80th-percentile error, suggesting that deterministic offset generation is less robust in difficult NLOS tail cases.

The NoEnv variant removes environmental conditioning while retaining diffusion, increasing the mean error from 0.82~m to 1.03~m in the lab scenario and from 1.08~m to 1.29~m in the office scenario. This accords with the role of $R(\mathbf{Y})$ in Corollary~\ref{cor:env_gain}. The NoDiff variant retains environmental input but replaces diffusion with a deterministic MLP, leading to mean errors of 0.91~m in the lab scenario and 1.23~m in the office scenario. These results match the decomposition in Eq.~\eqref{eq:full_bound}: removing $\mathbf{Y}$ increases the deterministic conditional-variance term, while replacing diffusion with a deterministic generator weakens the modeling of stochastic residual offsets. Thus, the ablation study provides empirical evidence for the two theoretical ingredients, namely geometry-dependent variance reduction and diffusion-based residual correction.

\begin{figure}
	\centering
	\subfigure[]{\includegraphics[width=0.225\textwidth]{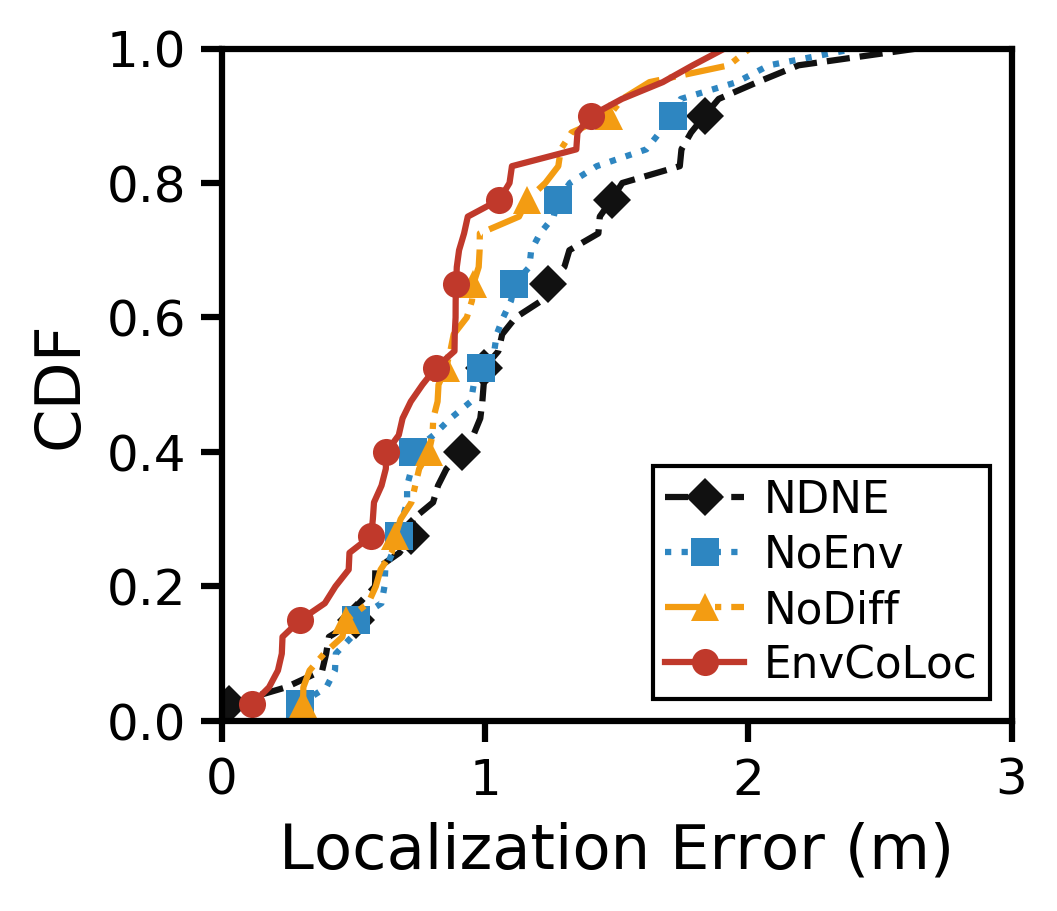}}
	\subfigure[]{\includegraphics[width=0.225\textwidth]{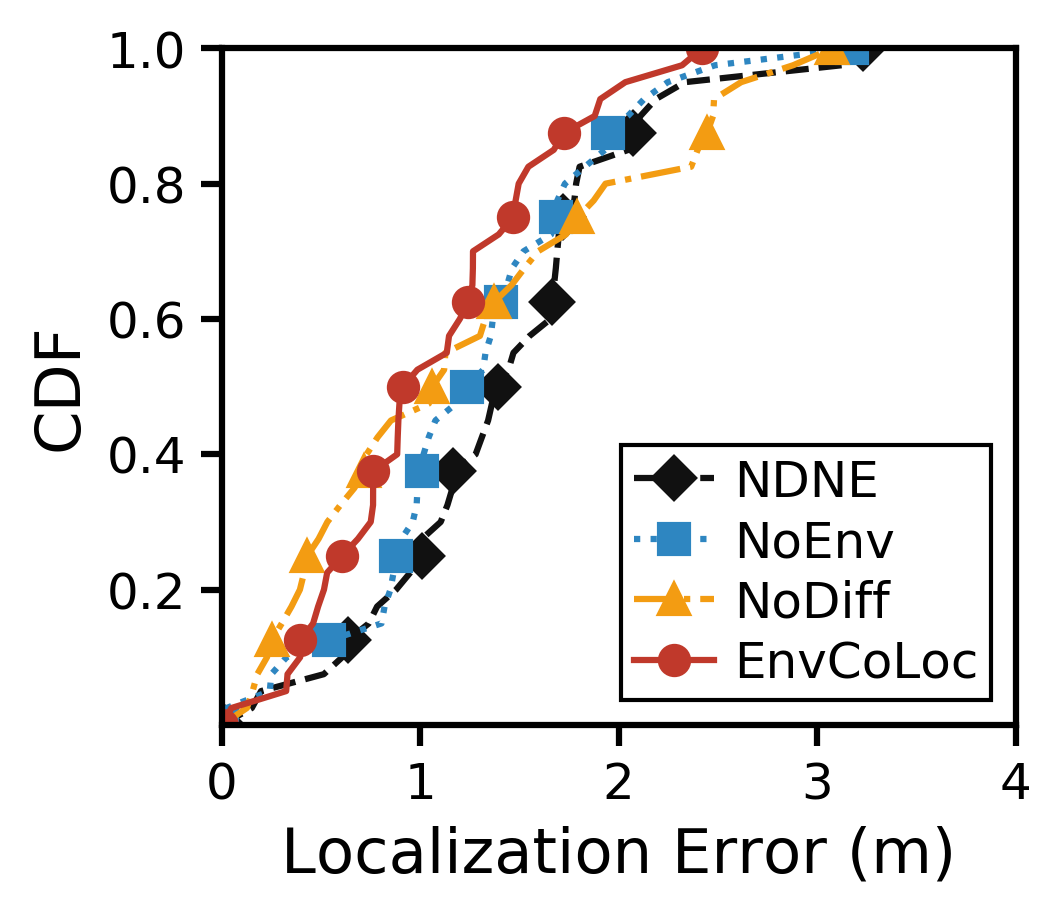}}
	\caption{CDF of localization errors for ablation variants with 10 support samples. (a) lab, \ac{los}. (b) office, \ac{nlos}.}~\label{fig:ablation_cdf}
\end{figure}

\begin{table}[t]
	\centering
	\caption{Ablation Study: Localization Error (m) with 10 Support Samples}
	\label{tab:ablation}
	\begin{tabular}{l|ccc|ccc}
		\toprule
		& \multicolumn{3}{c|}{lab (\ac{los})} & \multicolumn{3}{c}{office (\ac{nlos})} \\
		Variant & Mean & Med. & 80th & Mean & Med. & 80th \\
			\midrule
			NDNE       & 1.10 & 1.00 & 1.56 & 1.44 & 1.42 & 1.79 \\
			NoEnv      & 1.03 & 0.97 & 1.34 & 1.29 & 1.28 & 1.75 \\
			NoDiff     & 0.91 & 0.84 & 1.24 & 1.23 & 1.09 & 2.02 \\
			EnvCoLoc   & \textbf{0.82} & \textbf{0.79} & \textbf{1.09} & \textbf{1.08} & \textbf{0.95} & \textbf{1.51} \\
		\bottomrule
	\end{tabular}
\end{table}

\section{Conclusion}
\label{sec:conclusion}
In this paper, we proposed EnvCoLoc, an environment-conditioned meta-learning framework for data-efficient WiFi fingerprint localization that explicitly incorporates physical propagation geometry to mitigate cross-environment distribution shifts. By extracting material-aware descriptors from 3D point clouds within the first Fresnel zone, EnvCoLoc conditions a latent stochastic offset generator to produce environment-specific parameter corrections. This mechanism modulates a shared, meta-learned initialization, effectively decoupling the injection of geometric priors from subsequent wireless-data-driven gradient refinement. Real-world experiments using commodity WiFi devices validate this approach; with only 10 target support samples, EnvCoLoc reduced mean localization errors by 26.1\% in a LOS-dominant lab and 24.9\% in an NLOS-dominant office compared to the MetaLoc baseline. These results demonstrate that incorporating physical propagation geometry into meta-learning is an effective direction for reducing target-environment calibration effort in fingerprint-based indoor localization.

\bibliographystyle{IEEEtran}
\bibliography{EnvCoLoc}

\end{document}